\documentclass[11pt,onecolumn,oneside,peerreview]{IEEEtran}
\usepackage{cite}
\usepackage{graphicx}
\usepackage{amsmath}
\usepackage{times}
\usepackage{latexsym}
\usepackage{graphicx}
\usepackage{bm}
\usepackage{amssymb}
\usepackage[center]{caption2}
\usepackage{array}
\usepackage{setspace}
\usepackage{fancyhdr}
\usepackage{citesort}

\newtheorem{theorem}{Theorem}
\newtheorem{lemma}{Lemma}
\newtheorem{corollary}{Corollary}
\newtheorem{example}{Example}

\newcommand{\defeq}{\stackrel{\Delta}{=}}

\newcaptionstyle{mystyle2}{%
\captionlabel $.$ \, \doublespacing \captiontext \par}
\captionstyle{mystyle2}

\small

\begin{document}
\title{Ergodic Capacity Analysis of Amplify-and-Forward MIMO Dual-Hop Systems}
\author{
\begin{minipage}{0.98\columnwidth}
\vspace*{1cm}
\begin{center}
\authorblockN{Shi Jin\authorrefmark{1}\authorrefmark{3},
              Matthew R. McKay\authorrefmark{2},
              Caijun Zhong\authorrefmark{1},
              and Kai-Kit Wong\authorrefmark{1}\\
\vspace*{0.5cm} \small{
\authorblockA{\authorrefmark{1}Adastral Park Research Campus, University College London, United
Kingdom}\\
\authorblockA{\authorrefmark{2}Dept.\ of Electronic and Computer Engineering, Hong Kong University of Science \& Technology, Hong Kong} \\
\authorblockA{\authorrefmark{3}National Mobile Communications Research Laboratory,
Southeast University, Nanjing, China} } }
\end{center}
\end{minipage}
} \IEEEaftertitletext{\vspace{-0.75\baselineskip}}
\maketitle \vspace*{1cm}
\begin{abstract}
This paper presents an analytical characterization of the ergodic capacity of amplify-and-forward (AF) MIMO
dual-hop relay channels, assuming that the channel state information is available at the
destination terminal only. In contrast to prior results, our expressions apply for arbitrary
numbers of antennas and arbitrary relay configurations. We derive an expression for the exact ergodic capacity,
simplified closed-form expressions for the high SNR regime, and tight closed-form upper and lower bounds.
These results are made possible to employing recent tools from finite-dimensional random matrix theory to
derive new closed-form expressions for various statistical properties of the equivalent AF MIMO dual-hop relay
channel, such as the distribution of an unordered eigenvalue and certain random determinant properties. Based on
the analytical capacity expressions, we investigate the impact of the system and channel characteristics, such
as the antenna configuration and the relay power gain. We also demonstrate a number of interesting relationships
between the dual-hop AF MIMO relay channel and conventional point-to-point MIMO channels in various asymptotic regimes. 
\end{abstract}

\vskip 2ex \vspace*{0ex}
\begin{tabular}{ll}
{Corresponding Author\;:} {Shi Jin}\\
{Adastral Park Research Campus, University College London}\\
{Martlesham Heath, IP5 3RE, United
Kingdom}  \\
{E-mail: shijin@adastral.ucl.ac.uk, jinshi@seu.edu.cn} \\
\end{tabular}

\newpage

\section{Introduction}\label{sec:introduction}

The relay channel, first introduced in \cite{Meulen71,Cover79}, has been considered in recent
years as a means to improve the coverage and reliability, and to reduce the
interference in wireless
networks\cite{Laneman03,Sendonaris03A,Sendonaris03B,Laneman04,Gastpar05,Gastpar02,Kramer05,Madsen05,Farhadi08}.
Generally speaking, there are three main types of relaying
protocols: decode-and-forward (DF), compress-and-forward (CF), and
amplify-and-forward (AF). Of these protocols, the AF approach is the simplest
scheme, in which case the sources transmit messages to
the relays, which then simply scale their received signals according
to a power constraint and forward the scaled signals onto the
destinations.

Point-to-point multiple-input multiple-output (MIMO) communication systems have also been receiving considerable
attention in the last decade due to their potential for providing linear capacity growth and significant performance improvements over conventional
single-input single-output (SISO) systems \cite{Telatar99,Foschini98}. Recently, the application of MIMO
techniques in conjunction with relaying protocols has become a topic
of increasing interest as a means of achieving further performance improvements in wireless networks \cite{Borade03,Wittneben03,Nabar04,Wang05,Borade07}

In this paper we investigate the ergodic capacity of AF MIMO
dual-hop systems. This problem has been recently considered in
various settings.  In \cite{Bolcskei06}, the ergodic capacity of AF
MIMO dual-hop systems was examined for a large numbers of relay antennas $K$,
and was shown to scale with $\log K $. Asymptotic ergodic capacity
results were also obtained in \cite{Wagner07} by means of the
replica method from statistical physics. In
\cite{Morgenshtern06,Morgenshtern07}, the asymptotic network
capacity was examined as the number of source/desination antennas $M$ and relay antennas $K$ grew
large with a fixed-ratio $ K/M \to \beta$ using tools
from large-dimensional random matrix theory.  It was demonstrated
that for $\beta  \to \infty$, the relay network behaved equivalently
to a point-to-point MIMO link. The results of
\cite{Morgenshtern06,Morgenshtern07} were further elaborated in
\cite{Yeh07} where a general asymptotic ergodic capacity formula was
presented for multi-level AF relay networks. Recently, the asymptotic
mean and variance of the mutual information in correlated Rayleigh fading was studied in \cite{Wagner08}. All of these prior
capacity results, however, were derived by employing \emph{asymptotic methods} (i.e.\ by letting the system dimensions grow to infinity).  To the best of our knowledge, there appear to be no analytical ergodic capacity results which apply for AF MIMO dual hop systems with arbitrary finite antenna and relaying configurations.

In this paper we derive new exact analytical results, simple closed-form high SNR expressions, and tight
closed-form upper and lower bounds on the ergodic capacity of AF
MIMO dual-hop systems.  In contrast to previous results, our expressions apply for any finite number of
MIMO antennas and for arbitrary numbers of relay antennas. The results are based heavily on the theory of finite-dimensional random matrices. In
particular, our exact ergodic capacity results are based on a new exact expression which we derive for the exact unordered eigenvalue distribution of a certain product of finite-dimensional random matrices, corresponding to the equivalent cascaded AF MIMO relay channel. In prior work \cite{Morgenshtern07}, an asymptotic expression was obtained for this unordered eigenvalue density. However, that asymptotic result, which serves as an approximation for finite-dimensional systems, was rather complicated and required the numerical computation of a certain fixed-point equation. Our result, in contrast, is a simple exact closed-form expression, involving only standard functions which can be easily and efficiently evaluated.  In addition to the unordered eigenvalue distribution, we also present a number of new random determinant properties (such as the expected characteristic polynomial) of the equivalent cascaded AF MIMO relay channel. These results are subsequently employed to derive simplified closed-form expressions for the ergodic capacity in the high SNR regime, as well as tight upper and lower bounds. Again, these random determinant properties are exact closed-form analytical results which apply for arbitrary antenna and relaying configurations, and are expressed in terms of standard functions which are easy to compute.  As a by-product of these derivations, we also present some new \emph{unified} expressions for the expected characteristic polynomial and expected log-determinant of semi-correlated Wishart and pseudo-Wishart random matrices.

Based on our analytical expressions, we investigate the effect of the different system  and channel parameters on the ergodic capacity.  For example, we show that when either the number of source, destination, or relay antennas, or the the relay gain grows large, the AF MIMO dual-hop capacity admits a simple interpretation in terms of the ergodic capacity of conventional single-hop single-user MIMO channels.  In the high SNR regime, we present simple closed-form expressions for the key performance parameters---the high SNR slope and the high SNR power offset---which reveal the intuitive result that the multiplexing gain is determined by the minimum of the number of antennas at the source, destination, and relay, whereas the power offset is a more intricate function which depends on all three. For example, we show that by adding more antennas at the destination, whilst keeping the number of source and destination antennas fixed, may lead to a significant improvement in the high SNR power offset; however the relative gain becomes less significant as the initial number of destination antennas is increased. Our analytical expressions also reveal the interesting result that the ergodic capacity of AF MIMO dual-hop channels is upper bounded by the capacity of a SISO additive white Gaussian noise (AWGN) channel.

The remainder of this paper is organized as follows. Section
\ref{sec:model} presents the AF MIMO dual-hop system model under consideration.
Section \ref{sec:prelimibaries} presents our new random matrix
theory contributions, which are subsequently used to derive the exact, high SNR, and
upper and lower bound expressions for the ergodic capacity in Sections
\ref{sec:application} and \ref{sec:applicationbounds}. Section \ref{sec:conclusion} summarizes the main
results of the paper. All of the main mathematical proofs have been
placed in the Appendices.


\section{System Model}\label{sec:model}

We employ the same AF MIMO dual-hop system model as in
\cite{Morgenshtern06,Morgenshtern07}. In particular, suppose that
there are $n_s$ source antennas, $n_r$ relay antennas and $n_d$
destination antennas, which we represent by the $3$-tuple $\left(
{n_s ,n_r ,n_d } \right)$. All terminals operate in half-duplex
mode, and as such communication occurs from source to relay and from
relay to destination in two separate time slots. It is assumed that
there is no direct communication link between the source and
destination, as sketched in Fig. \ref{fig:fig0}. The end-to-end
input-output relation of this channel is then given by
\begin{align}\label{eq:signalmodel}
{\bf{y}} = {\bf{H}}_2 {\bf{FH}}_1 {\bf{s}} + {\bf{H}}_2
{\bf{Fn}}_{n_r} + {\bf{n}}_{n_d}
\end{align}
where $ {\bf{s}}$ is the transmit symbol vector, $ {\bf{n}}_{n_r} $
and $ {\bf{n}}_{n_d} $ are the relay and destination noise vectors
respectively, $ {\bf{F}} =\sqrt {\alpha /\left( {n_r \left( {1 +
\rho } \right)} \right)} \mathbf{I}_{n_r}  $ ($\alpha$ corresponds to the overall power
gain of the relay terminal) is the forwarding matrix at the relay
terminal which simply forwards scaled versions of its received
signals, and $ {\bf{H}}_1  \in \mathcal{C}^{n_r  \times n_s}$ and $
{\bf{H}}_2  \in \mathcal{C}^{n_d  \times n_r}$ denote the channel
matrices of the first hop and the second hop respectively, where
their entries are assumed to be zero mean circular symmetric complex
Gaussian (ZMCSCG) random variables of unit variance. The input
symbols are chosen to be independent and identically distributed (i.i.d.) ZMCSCGs and the per antenna power is
assumed to be $ \rho /{n_s} $, i.e., $ E\left\{ {{\bf{ss}}^\dag }
\right\} = \left( {\rho /{n_s}} \right){\bf{I}}_{n_s} $. The
additive noise at the relay and destination are assumed to be white
in both space and time and are modeled as ZMCSCG with unit variance,
i.e., $ E\left\{ {{\bf{n}}_{n_r} {\bf{n}}_{n_r}^\dag  } \right\} =
{\bf{I}}_{n_r}$ and $ E\left\{ {{\bf{n}}_{n_d} {\bf{n}}_{n_d}^\dag }
\right\} = {\bf{I}}_{n_d}$.  We assume that the source and relay
have no channel state information (CSI), and that the destination
has perfect knowledge of both $\mathbf{H}_2$ and
$\mathbf{H}_2\mathbf{H}_1$.

\begin{figure}
\centering
\includegraphics[scale=0.5]{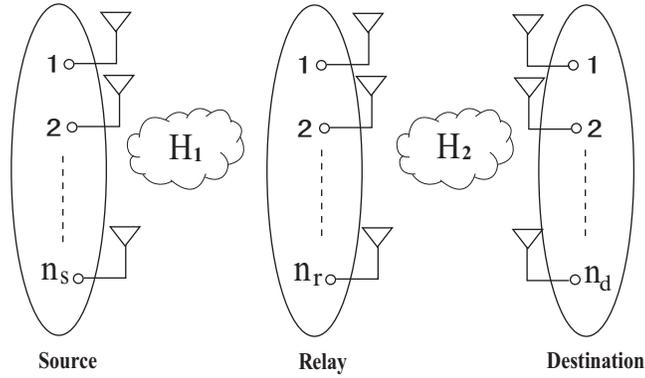}
\captionstyle{mystyle2}\caption{Schematic diagram of a MIMO dual-hop
system, where there is no direct link between source and
destination.} \label{fig:fig0}
\end{figure}

The ergodic capacity (in b/s/Hz) of the AF MIMO dual-hop
system described above can be written as \cite{Morgenshtern06,Morgenshtern07,Wagner07}
\begin{align}\label{eq:capa1}
C = \frac{1}{2}E\left\{ {\log_2 \det \left( {{\bf{I}} + {\bf{R}}_s
{\bf{R}}_n^{ - 1} } \right)} \right\}
\end{align}
where ${\bf{R}}_s$ and ${\bf{R}}_n$ are $n_d \times n_d$ matrices
given by
\begin{align}\label{eq:rs}
{\bf{R}}_s  = \frac{{\rho a}}{{n_s }}{\bf{H}}_2 {\bf{H}}_1
{\bf{H}}_1^\dag  {\bf{H}}_2^\dag
\end{align}
and
\begin{align}\label{eq:rn}
{\bf{R}}_n  = {\bf{I}}_{n_d}  + a{\bf{H}}_2 {\bf{H}}_2^\dag
\end{align}
respectively, with
\begin{align}  \label{eq:aDefn}
a = \frac{\alpha }{{n_r \left( {1 + \rho } \right)}}.
\end{align}
Using the identity
\begin{align} \label{eq:DetProp}
\det \left( {{\bf{I}} + {\bf{AB}}} \right) = \det \left( {{\bf{I}} +
{\bf{BA}}} \right),
\end{align}
(\ref{eq:capa1}) can be alternatively expressed as follows
\begin{align}\label{eq:ergodic1}
C\left( \rho  \right) = \frac{1}{2}E\left\{ {\log_2 \det \left(
{{\bf{I}}_{n_s}  + \frac{{\rho a}}{{n_s }}{\bf{H}}_1^\dag
{\bf{H}}_2^\dag {\bf{R}}_n ^{ - 1} {\bf{H}}_2 {\bf{H}}_1 } \right)}
\right\} \; .
\end{align}
Next, we utilize the singular value decomposition to write $
{\bf{H}}_2  = {\bf{U}}_2 {\bf{D}}_2 {\bf{V}}_2^\dag$, where
\begin{align}
{\bf{D}}_2  = {\rm diag}\left\{ {\lambda _1 , \ldots ,\lambda _{\min
\left( {n_d ,n_r } \right)} } \right\}
\end{align}
is an $n_d \times n_r$ diagonal matrix, with diagonal elements
pertaining to the increasing ordered singular values, and $
{\bf{U}}_2 \in \mathcal{C}^{n_d \times n_d} $ and ${\bf{V}}_2  \in
\mathcal{C}^{n_r \times n_r}$ are unitary matrices containing the
respective eigenvectors.  Since $\mathbf{H}_1$ is invariant under
left and right unitary transformation, the ergodic capacity in
(\ref{eq:ergodic1}) can be further simplified as
\begin{align}\label{eq:ergodic2}
C\left( \rho  \right) = \frac{1}{2}E\left\{ {\log_2 \det \left(
{{\bf{I}}_{n_r}  + \frac{{\rho a}}{{n_s }}{\bf{H}}_1^\dag {\bf{\Psi
H}}_1 } \right)} \right\}
\end{align}
where
\begin{align}\label{eq:psiexp}
{\bf{\Psi }} = \left\{ {\begin{array}{*{20}c}
   { {\rm diag}\left\{ {\frac{{\lambda _1^2 }}{{1 + a\lambda _1^2 }}, \ldots ,\frac{{\lambda _{n_r}^2 }}{{1 + a\lambda _{n_r}^2 }}} \right\},} & {n_r \le n_d},  \\
   { {\rm diag}\left\{ {\frac{{\lambda _1^2 }}{{1 + a\lambda _1^2 }}, \ldots ,\frac{{\lambda _{n_d}^2 }}{{1 + a\lambda _{n_d}^2 }},\underbrace {0, \ldots ,0}_{n_r - n_d}} \right\},} & {n_r > n_d.}  \\
\end{array}} \right.
\end{align}
It is then easily established that
\begin{align}\label{eq:ergcapa}
C\left( \rho  \right) = \frac{1} {2}E\left\{ {\log _2 \det \left(
{{\mathbf{I}}_{n_s }  + \frac{{\rho a}}{{n_s }}{\mathbf{\tilde
H}}_1^\dag  {\mathbf{L\tilde H}}_1 } \right)} \right\}
\end{align}
where ${\mathbf{\tilde H}}_1^\dag   \sim \mathcal{CN}_{n_s ,q}
\left( {{\mathbf{0}},{\mathbf{I}}_{n_s }  \otimes {\mathbf{I}}_q }
\right)$, with $ q = \min \left( {n_d ,n_r } \right)$, and
\begin{align}\label{eq:Ldefinition}
{\mathbf{L}} = {\rm diag}\left\{ {\lambda _i^2 /\left( {1 + a\lambda
_i^2 } \right)} \right\}_{i = 1}^q  \; .
\end{align}
Equivalently, we can now write
\begin{align}\label{eq:ergodic3}
C\left( \rho  \right) = \frac{ s}{2}\int_0^\infty  {\log_2 \left( {1
+ \frac{{\rho a}}{n_s}\lambda } \right)f_\lambda  \left( \lambda
\right)d\lambda }
\end{align}
where $ s = \min \left( {n_s ,q} \right)$, $\lambda$ denotes an
unordered eigenvalue of the random matrix ${\mathbf{\tilde
H}}_1^\dag  {\mathbf{L\tilde H}}_1 $, and $f_\lambda  \left( \cdot
\right)$ denotes the corresponding probability density function (p.d.f.). Although the distribution
of $\lambda$ has been well-studied in the asymptotic antenna regime
\cite{Morgenshtern06,Morgenshtern07}, currently there are no exact
closed-form expressions for $f_\lambda(\cdot)$ which apply for
arbitrary finite-antenna systems.

\section{New Random Matrix Theory Results}\label{sec:prelimibaries}

In this section, we derive a new exact closed-form expression for
the unordered eigenvalue distribution $f_{\lambda}(\cdot)$ of the
random matrix ${\mathbf{\tilde H}}_1^\dag {\mathbf{L\tilde H}}_1 $.
We also present a number of other key results, such as random
determinant properties, which will prove useful in subsequent
derivations. It is convenient to define the following notation:
$\alpha _i  = \lambda _i^2$, $\beta _i = \lambda _i^2 /\left( {1 +
a\lambda _i^2 } \right)$ $ ( i = 1, \ldots ,q )$, and $ p = \max
\left( {n_d,n_r } \right) $.


To derive the unordered eigenvalue distribution $f_\lambda(\cdot)$,
we first need to establish some key preliminary results, as given
below.


\begin{lemma}\label{unorderedpdf}
The marginal p.d.f.\ of an unordered eigenvalue $\lambda$ of
${\mathbf{\tilde H}}_1^\dag  {\mathbf{L\tilde H}}_1 $, conditioned
on $\mathbf{L}$, is given by
\begin{align}
f_{\left. \lambda  \right|{\bf{L}}} \left( \lambda  \right) =
\frac{1}{{s \prod_{i<j}^q ( \beta_j - \beta_i ) }} {\sum\limits_{l =
1}^q \sum\limits_{k = q - s + 1}^q  {\frac{{\lambda ^{n_s  + k - q -
1} e^{ - \lambda /\beta _l } \beta _l^{q - n_s  - 1} }}{{\Gamma
\left( {n_s - q + k} \right)}} {D}_{l,k} } }
\end{align}
where $ {D}_{l,k}  $ is the $\left( {l,k} \right)$th cofactor of a
$q \times q$ matrix ${{\bf{D}}}$ whose $\left( {m,n} \right)$th
entry is
\begin{align}
\left\{ {{\bf{D}}} \right\}_{m,n}  = \beta _m^{n - 1}.
\end{align}
\end{lemma}
\begin{proof}
See Appendix \ref{sec:Proof_unorderedcond}.
\end{proof}
This lemma presents a new expression for the unordered eigenvalue
distribution of a complex semi-correlated central Wishart matrix. In
prior work \cite{Alfano04}, two separate alternative expressions for
this p.d.f.\ were obtained for the specific scenarios $n_s \leq q$
and $n_s > q$ respectively; the latter case\footnote{For this case
($n_s > q$), the random matrix ${\mathbf{\tilde H}}_1^\dag
{\mathbf{L\tilde H}}_1 $ has reduced rank and the corresponding
distribution, conditioned on $\mathbf{L}$, is commonly referred to
as \emph{pseudo}-Wishart \cite{Mallik03}.} being a complicated expression in terms
of determinants with entries depending on the inverse of a certain
Vandermonde matrix. Here, \textit{Lemma \ref{unorderedpdf}} presents
a simpler and more computationally-efficient \emph{unified}
expression, which applies for arbitrary $n_s$ and $q$.


To remove the conditioning on $\mathbf{L}$ in \emph{Lemma
\ref{unorderedpdf}}, it is necessary to establish a closed-form
expression for the joint p.d.f.\ of $\beta _1 , \cdots ,\beta _q$.
We will also require the p.d.f.\ of an arbitrarily selected $\beta
\in \{ \beta _1 , \cdots ,\beta _q \}$. These results are given in
the following lemma.
\begin{lemma}\label{jointpdf_beta}
The joint p.d.f. of $ \{0 \leq \beta _1  < \cdots < \beta _q \leq
1/a \}$ is given by
\begin{align} \label{eq:jointPDFg}
f(\beta_1, \ldots, \beta_q) = \mathcal{K}  \prod_{i<j}^q ( \beta_j -
\beta_i )^2 \prod_{i=1}^{q} \frac{\beta_i^{p - q}
e^{-\frac{\beta_i}{1-a \beta_i}} }{ (1- a \beta_i)^{p + q} }
\end{align}
where
\begin{align}
\mathcal{K} = \left( \prod\nolimits_{i = 1}^q {\Gamma \left( {q - i
+ 1} \right)\Gamma \left( {p - i + 1} \right)} \right)^{-1} \; .
\end{align}
The p.d.f. of an unordered (randomly-selected) $\beta \in \{\beta
_1, \cdots , \beta _q \}$ is given by
\begin{align} \label{eq:Unordered_Beta}
f\left( \beta  \right) = \frac{1}{q}\sum\limits_{i = 0}^{q - 1}
{\sum\limits_{j = 0}^i {\sum\limits_{l = 0}^{2j}
{\frac{{\mathcal{A}\left( {i,j,l,p,q } \right)\beta ^{p - q + l}
}}{{\left( {1 - a\beta } \right)^{p - q + l + 2} }}\exp \left( { -
\frac{\beta }{{1 - a\beta }}} \right)} } }
\end{align}
where
\begin{align}
\mathcal{A}\left( {i,j,l,\kappa _1,\kappa _2 } \right) =
\frac{\left( { - 1} \right)^l \binom{2i-2j}{i-j} \binom{2j + 2\kappa
_1 - 2\kappa _2}{2j - l} ( 2j )!  }{{2^{2i - l} \, ({\kappa _1 -
\kappa _2 + j} )! \, j! }} \; .
%
%
\end{align}
\end{lemma}
\begin{proof}
See Appendix \ref{sec:Proof_Jointpdf_beta}.
\end{proof}

Having established the results in \emph{Lemma \ref{unorderedpdf}}
and \emph{Lemma \ref{jointpdf_beta}}, we are now ready to derive the
desired unconditional unordered eigenvalue distribution
$f_\lambda(\cdot)$, as given below.

\begin{theorem}\label{final_unorderedpdf}
The marginal p.d.f.\ of an unordered eigenvalue $\lambda$ of
${\mathbf{\tilde H}}_1^\dag  {\mathbf{L\tilde H}}_1 $ is given by
\begin{align}\label{eq:unorderpdffinal}
f_\lambda  \left( \lambda  \right) &= \frac{{2e^{ - \lambda
a}\mathcal{K}}}{s}\sum\limits_{l = 1}^q\sum\limits_{k = q - s + 1}^q
{ {\sum\limits_{i = 0}^{q + n_s - l} } } \frac{ \binom{q + n_s -
l}{i}  {a^{q + n_s  - l - i} }}{{\Gamma \left( {n_s  - q + k}
\right)}}  \lambda ^{\left( {2n_s + 2k + p - q - i - 3} \right)/2}
K_{p + q - i - 1} \left( {2\sqrt \lambda  } \right)
 G_{l,k}
\end{align}
where $K_v \left( \cdot \right)$ is the modified Bessel function of
the second kind and $ G_{l,k} $ is the $\left( {l,k} \right)$th
cofactor of a $q \times q$ matrix $\bf{G}$ whose $\left( {m,n}
\right)$th entry is
\begin{align}
\left\{ {{\bf{G}}} \right\}_{m,n}  = a^{q - p - m - n + 1}\Gamma
\left( {p - q + m + n - 1} \right)U\left( {p - q + m + n - 1,p +
q,1/a} \right) \;
\end{align}
with $ U\left( { \cdot , \cdot , \cdot } \right)$ denoting the
confluent hypergeometric function of the second kind \cite[Eq.
9.211.4]{Gradshteyn00}.
\end{theorem}
\begin{proof}
See Appendix \ref{sec:Proof_finalunorderedpdf}.
\end{proof}
We note that an asymptotic expression for $f_{\lambda}(\cdot)$ has
been considered previously in \cite{Morgenshtern07}, based on
large-dimensional random matrix theory. However, that asymptotic
p.d.f.\ result, which serves as an approximation for
finite-dimensional systems, is not in closed-form, requiring the
numerical computation of a certain fixed-point equation. Indeed, to
further facilitate computation of the asymptotic eigenvalue p.d.f.\
in \cite{Morgenshtern07}, an algorithmic approach with certain
heuristic elements was also presented. Our result in \emph{Theorem
\ref{final_unorderedpdf}}, in contrast, gives the \emph{exact}
eigenvalue p.d.f.\ which applies for arbitrary finite system
dimensions, and is presented in a simple closed-form involving only
standard functions which can be easily and efficiently evaluated. In
the following section, this result will be employed to evaluate the
ergodic capacity of AF MIMO dual-hop channels.


\begin{corollary}\label{pdf_special1}
For the special case $\left( {1,1,1} \right)$, the unordered
eigenvalue p.d.f.\ (\ref{eq:unorderpdffinal}) reduces to
\begin{align}
f_\lambda ^{\left( {1,1,1} \right)}(\lambda)  = 2e^{ -
\frac{{\lambda \alpha }}{{1 + \rho }}} \left[ {\left( {\frac{\alpha
}{{1 + \rho }}} \right)\sqrt \lambda  K_1 \left( {2\sqrt \lambda  }
\right) + K_0 \left( {2\sqrt \lambda  } \right)} \right] \; .
\end{align}
\end{corollary}
\begin{proof}
The proof is straightforward and is omitted.
\end{proof}
We note that this special case has also been derived previously in
\cite{Hasna04}.

\begin{corollary}\label{final_unorderedproduct}
Let ${\bf{\tilde L}} = {\rm diag}\left\{ {\lambda _i^2 } \right\}_{i
= 1}^q $. Then, the marginal p.d.f.\ of an unordered eigenvalue
$\lambda$ of ${\mathbf{\tilde H}}_1^\dag  {\mathbf{\tilde L \tilde
H}}_1 $ is given by
\begin{align}\label{eq:unorderedproduct}
\tilde f_\lambda  \left( \lambda  \right) =
\frac{{2\mathcal{K}}}{s}\sum\limits_{l = 1}^q\sum\limits_{k = q - s
+ 1}^q { {\frac{{\lambda ^{\left( {n_s  + 2k + p + l - 2q - 3}
\right)/2} }}{{\Gamma \left( {n_s  - q + k} \right)}}} } K_{p - n_s
+ l - 1} \left( {2\sqrt \lambda  } \right) {{\bar G}}_{l,k}
\end{align}
where ${{{\bar G}}_{l,k} }$ is the $\left( {l,k} \right)$th cofactor
of a $q \times q$ matrix ${\bf{\bar G}}$ whose $\left( {m,n}
\right)$th entry is
\begin{align} \label{eq:GBar}
\left\{ {{\bf{\bar G}} } \right\}_{m,n}  = \Gamma \left( {p - q + m + n - 1} \right).
\end{align}
\end{corollary}
\begin{proof}
The result is obtained by taking the limit as $a \to 0$ in
(\ref{eq:unorderpdffinal}).
\end{proof}
This result will be used to study the capacity of AF MIMO dual-hop
channels in the high SNR regime. It is also worth noting that
(\ref{eq:unorderedproduct}) can be applied to the ergodic capacity
analysis of Rayleigh-product MIMO channels \cite{Yang07,Jin08}.

Fig. \ref{fig:fig3} compares the analytical result presented in
\textit{Theorem \ref{final_unorderedpdf}} with Monte Carlo
simulations. We plot the p.d.f. of the unordered eigenvalue $\lambda$
with system configuration $(2, 3, 4)$. The simulated p.d.f. curve is
based on 100,000 channel realizations. The figure shows that
the analytical result is in agreement with the simulations.

Fig. \ref{fig:fig4} shows the analytical result presented in
\textit{Theorem \ref{final_unorderedpdf}} and \textit{Corollary
\ref{final_unorderedproduct}}. The curves corresponding to $\rho = 0{~}
{\rm dB}$, $\rho = 10{~} {\rm dB}$, and $\rho = 20{~} {\rm dB}$ are generated using
(\ref{eq:unorderedproduct}) while the ``Rayleigh Product" curve is
generated using (\ref{eq:unorderedproduct}). We can see that the
exact unordered eigenvalue distribution converges to the
unordered eigenvalue distribution of the Rayleigh product channel as $a \to \infty$, as expected.

Fig. \ref{fig:fig12} compares our exact unordered eigenvalue distribution, based on (\ref{eq:unorderpdffinal}), with the corresponding asymptotic eigenvalue distribution presented in \cite{Morgenshtern07}, for the random matrix ${{\bf{\tilde H}}_1^\dag {\bf{ L\tilde H}}_1
}/(n_sn_r)$ with different system configurations. We use the same
simulation parameters as in \cite[Fig. 5 (a)]{Morgenshtern07}, setting
$a = 1/ n_r$ and $n_r / n_s = 1 / 2$. We clearly see the convergence of the exact and asymptotic p.d.f.s as the numbers of antennas become large (eg.\ the $(16,8,16)$ scenario), however when the systems dimensions are not so large (eg.\ the $(2,1,2)$ and $(4,2,4$) scenarios), the asymptotic eigenvalue p.d.f.\ exhibits noticeable inaccuracies with respect to our new exact result in (\ref{eq:unorderpdffinal}).


\begin{figure}
\centering
\includegraphics[scale=0.8]{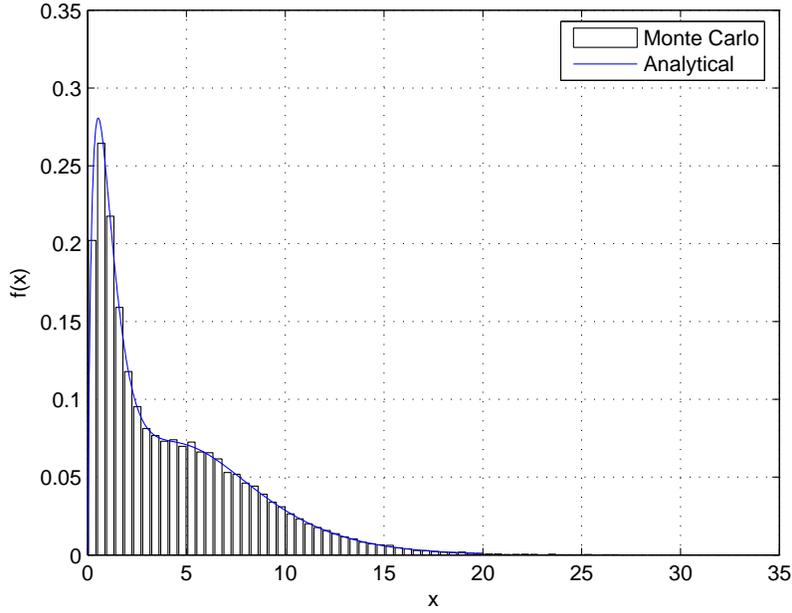}
\captionstyle{mystyle2}\caption{Comparison of the analytical and Monte Carlo-simulated unordered eigenvalue
p.d.f.\ of ${\mathbf{\tilde H}}_1^\dag
{\mathbf{L\tilde H}}_1 $.  Results are shown for $(2, 3, 4)$ system configuration, with
$\alpha = 2$.} \label{fig:fig3}
\end{figure}

\begin{figure}
\centering
\includegraphics[scale=0.8]{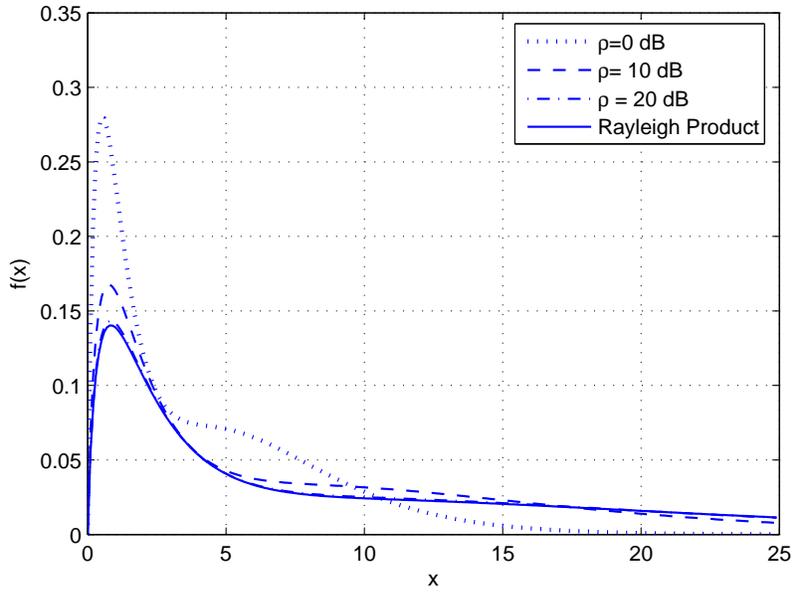}
\captionstyle{mystyle2}\caption{Comparison of the analytical unordered eigenvalue p.d.f.\ of ${{\bf{\tilde H}}_1^\dag {\bf{ L\tilde
H}}_1 }$ and ${{\bf{\tilde H}}_1^\dag {\bf{ \tilde L\tilde H}}_1 }$ for different
$\rho$.  Results are shown for a
$(2, 3, 4)$ system configuration, with $\alpha = 2$.}
\label{fig:fig4}
\end{figure}

\begin{figure}
\centering
\includegraphics[scale=0.8]{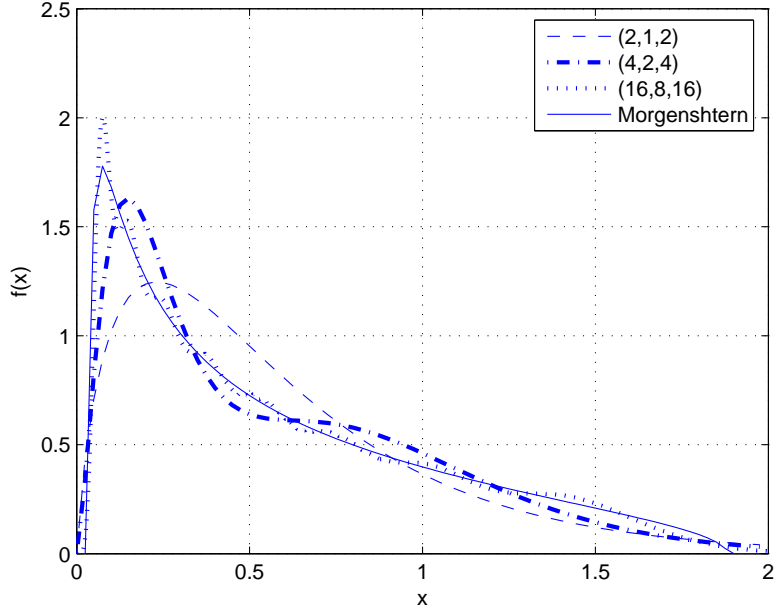}
\captionstyle{mystyle2}\caption{Comparison of the analytical unordered eigenvalue p.d.f.\ of ${{\bf{\tilde H}}_1^\dag {\bf{ L\tilde
H}}_1 }/(n_sn_r)$ for different system configurations.  Results are shown for $ a =
1/n_r$ and $n_r / n_s = 1 / 2$.}
\label{fig:fig12}
\end{figure}

The following theorems present new closed-form random determinant
properties, involving the random matrix ${\mathbf{\tilde H}}_1^\dag
{\mathbf{ L \tilde H}}_1 $.  These results will be applied to
derive tight bounds on the ergodic capacity.


\begin{lemma}\label{expdet}
The expected determinant of ${\bf{I}}_{n_s }  + \left( {\rho a/n_s }
\right){\bf{\tilde H}}_1^\dag  {\bf{L\tilde H}}_1$, conditioned on
$\mathbf{L}$, is given by
\begin{align}\label{eq:expdet}
E\left\{ {\left. {\det \left( {{\bf{I}}_{n_s }  + \frac{{\rho
a}}{{n_s }}{\bf{\tilde H}}_1^\dag  {\bf{L\tilde H}}_1 } \right)}
\right|{\bf{L}}} \right\} = \frac{{\det \left( {\bf{\Delta }}
\right)}}{ {\prod\nolimits_{i < j}^q {\left( {\beta _j  - \beta _i }
\right)} }}
\end{align}
where $ {\bf{\Delta }}$ is a $q \times q$ matrix with entries\footnote{When $q < n_s$, $\left\{ {\bf{\Delta }} \right\}_{m,n} = {\beta _m^{n - 1} \left( {1 + \frac{{\rho a}}{{n_s }}\beta _m \left( {n_s  - q + n} \right)} \right)}$.}
\begin{align}
\left\{ {\bf{\Delta }} \right\}_{m,n}  = \left\{
{\begin{array}{*{20}c}
   {\beta _m^{n - 1} }, & {n \le q - n_s, }  \\
   {\beta _m^{n - 1} \left( {1 + \frac{{\rho a}}{{n_s }}\beta _m \left( {n_s  - q + n} \right)} \right)}, & {n > q - n_s. }  \\
\end{array}} \right. \; \;
\end{align}
\end{lemma}
\begin{proof}
See Appendix \ref{sec:Proof_expdet}.
\end{proof}
This theorem presents a new expression for the expected
characteristic polynomial of a complex semi-correlated central
Wishart matrix. In prior work \cite{Shin03,Zhang05},  alternative
expressions were obtained via a different approach (i.e.\ by
exploiting a classical characteristic polynomial expansion for the
determinant).  Those results, however, involved summations over
subsets of numbers, with each term involving determinants of
partitioned matrices.  In contrast, our result in \textit{Lemma
\ref{unorderedpdf}} is more computationally-efficient, involving
only a single determinant with simple entries. Moreover, it is more
amenable to the further analysis in this paper, leading to the
following important theorem.


\begin{theorem}\label{expdetNew}
The unconditional expected determinant of ${\bf{I}}_{n_s } + \left(
{\rho a/n_s } \right){\bf{\tilde H}}_1^\dag  {\bf{L\tilde H}}_1$ is
given by
\begin{align}
E\left\{ \det \left( {{\bf{I}}_{n_s }  + \frac{{\rho a}}{{n_s
}}{\bf{\tilde H}}_1^\dag  {\bf{L\tilde H}}_1 } \right) \right\} =
\mathcal{K}\det \left( {\bf{\bar {\Xi} }} \right)
\end{align}
where $ {\bf{\bar {\Xi} }}$ is a $q \times q$ matrix with entries
\begin{align} \label{eq:XiBar}
\left\{ {\bf{\bar {\Xi} }} \right\}_{m,n}  = \left\{
{\begin{array}{*{20}c}
   a^{1 - \tau } \vartheta_{\tau-1}( a ), & {n \le q - n_s }  \\
   a^{1 - \tau } \left( \vartheta_{\tau-1}( a ) + \frac{\rho }{{n_s }}  \left( {n_s  - q + n } \right)
\vartheta_{\tau}( a ) \right), & {n > q - n_s }  \\
\end{array}} \right.
\end{align}
with $ \tau  = p - q + m + n $, and
\begin{align} \label{eq:VarThetaDefn}
\vartheta_\tau( a ) = \Gamma \left( {\tau } \right)U\left( {\tau,p +
q,1/a} \right) \; .
\end{align}
%


\end{theorem}
\begin{proof}
Utilizing \textit{Lemma \ref{expdet}}, \cite[Lemma 2]{Shin06} and
(\ref{eq:ykij1}) yields the desired result.
\end{proof}

\begin{lemma}\label{lnexpdet}
Let
\begin{align}
{\bf{\Phi }} = \left\{ {\begin{array}{*{20}c}
   {{\bf{\tilde H}}_1^\dag  {\bf{L\tilde H}}_1 }, & {q \ge n_s, }  \\
   {{\bf{L\tilde H}}_1 {\bf{\tilde H}}_1^\dag  }, & {q < n_s. }  \\
\end{array}} \right. \;
\end{align}
The expected log-determinant of ${\bf{\Phi }}$, conditioned on
$\mathbf{L}$, is given by
\begin{align}\label{eq:lnexpdet}
E\left\{ {\left. {\ln \det \left( {\bf{\Phi }} \right)}
\right|{\bf{L}}} \right\} =
 \sum\limits_{k = 1}^s {\psi
\left( {n_s-s +k} \right)} + \frac{\sum \limits_{k = q - s + 1}^q
{\det \left( {{\bf{Y}}_k } \right)}}{{\prod\nolimits_{i < j}^q
{\left( {\beta _j  - \beta _i } \right)} }}
\end{align}
where $\psi \left( \cdot \right)$ is the digamma
function\cite{Gradshteyn00}, and
%
$ {{\bf{Y}}_k } $ is a $q \times q$ matrix
with entries
\begin{align}
\left\{ {{\bf{Y}}_k } \right\}_{m,n}  = \left\{
{\begin{array}{*{20}c}
   {\beta _m^{n - 1} }, & {n \ne k,}  \\
   {\beta _m^{n - 1} \ln \beta _m }, & {n = k.}  \\
\end{array}} \right.
\end{align}
When $q = s$, (\ref{eq:lnexpdet}) reduces to
\begin{align}\label{eq:lndetsp1}
E\left\{ {\left. {\ln \det \left( {\bf{\Phi }} \right)}
\right|{\bf{L}}} \right\} = \sum\limits_{k = 1}^s {\psi \left( {n_s
- s + k} \right)}  + \ln \det \left( {\bf{L}} \right) \; .
\end{align}

\end{lemma}
\begin{proof}
See Appendix \ref{sec:Proof_lnexpdet}.
\end{proof}
We note that the above expected natural logarithm of the determinant
for $q \ge n_s$ has been investigated in \cite{Lozano05}, where the
derived expression is rather complicated, involving summations of
determinants whose elements are in terms of the inverse of a certain
Vandermonde matrix. We also note the $q < n_s$ and $q = n_s=s$ cases
have been considered in \cite{Zhang05,Grant02}. Our result, in
contrast, gives a simple \textit{unified} expression which embodies
all of these cases. Moreover, based on \textit{Lemma
\ref{lnexpdet}}, we obtain the following important theorem.

\begin{theorem}\label{lnexpdetNew}
The unconditional expected log-determinant of ${\bf{\Phi }}$ is
given by
\begin{align}\label{eq:lnexpdetNew1}
E\left\{ { {\ln \det \left( {\bf{\Phi }} \right)}   } \right\} =
\sum\limits_{k = 1}^s {\psi \left( {n_s -s + k } \right)} +
\mathcal{K}\sum\limits_{k = q - s  + 1}^q {\det \left( {{\bf{W}}_k }
\right)}
\end{align}
where ${{\bf{W}}_k}$ is a $q \times q$ matrix with entries
\begin{align} \label{eq:WkDefn}
\left\{ {{\bf{W}}_k } \right\}_{m,n}  = \left\{
{\begin{array}{*{20}c}
a^{1 - \tau } \vartheta_{\tau-1}( a ), & {n \ne k}  \\
\varsigma_{m+n}(a), & {n = k}  \\
\end{array}} \right.
\end{align}
where
$\tau$ and $\vartheta_{\tau-1}(\cdot)$ are defined as in
(\ref{eq:VarThetaDefn}), and
\begin{align}
\varsigma_{t}(a) = \sum\limits_{i = 0}^{2q - t} {a^{2q - t - i} \Gamma \left( {p + q - i - 1} \right) \binom{2q-t}{i}
} \left( {\psi \left( {p + q - i - 1} \right) - \sum\limits_{l = 0}^{p + q - i - 2} {g_l \left( {\frac{1}{a}} \right)} } \right)
\end{align}
where 
$g_l (\cdot)$ denotes the auxiliary function
\begin{align} \label{eq:gDefn}
g_l (x) = e^x E_{l+1} (x) \;
\end{align}
with $E_{l+1} \left( \cdot \right)$ denoting the exponential
integral function of order $l+1$.

When $q = s$, (\ref{eq:lnexpdetNew1}) reduces to
\begin{align}\label{eq:lnexpdetNew2}
E\left\{ { {\ln \det \left( {\bf{\Phi }} \right)}   } \right\} &=
\sum\limits_{k = 1}^s {\psi \left( {n_s  - s + k} \right)}
+\sum\limits_{i = 0}^{q - 1} {\sum\limits_{j = 0}^i {\sum\limits_{l
= 0}^{2j} {\sum\limits_{k = 0}^{2q - l - 2} {  { \binom{2q - l -
2}{k}
%
%
\mathcal{A}\left( {i,j,l,p,q} \right)} } } } } \nonumber\\
& \hspace*{1cm} { \times a^{2q - l - 2 - k} \Gamma \left( {p + q - k
- 1} \right) \biggl( {\psi \left( {p + q - k - 1} \right) -
\sum\limits_{m = 0}^{p + q - k - 2} {g_{m} \left( 1/a \right)} }
\biggr)} \; .
\end{align}
\end{theorem}
\begin{proof}
See Appendix \ref{sec:Proof_lnexpdetNew}.
\end{proof}

\section{Ergodic Capacity Analysis}\label{sec:application}


In this section we present new analytical expressions for the
ergodic capacity of AF MIMO dual-hop systems.

\subsection{Exact Expression for Ergodic Capacity}

%

Substituting (\ref{eq:unorderpdffinal}) into (\ref{eq:ergodic3}) we obtain
\begin{align}\label{eq:blah}
C\left( \rho  \right) &= \mathcal{K} \sum\limits_{l = 1}^q
\sum\limits_{k = q - s + 1}^q {{\sum\limits_{i = 0}^{q + n_s - l}}}
\frac{ \binom{q + n_s - l}{i} {a^{q + n_s  - l - i} } }{{\Gamma
\left( {n_s - q + k} \right)}} {{{G}}_{l,k} } \mathcal{J}_{i,k}
\end{align}
where
\begin{align} \label{eq:intExpression}
\mathcal{J}_{i,k} &= \int_0^\infty  \log_2 \left( {1 + \frac{{\rho
a}}{n_s}\lambda } \right) e^{ - \lambda a} \lambda^{\left( {2n_s +
2k + p - q - i - 3} \right)/2} K_{p + q - i - 1} \left( {2\sqrt
\lambda } \right) d\lambda  \; .
\end{align}
The integral in (\ref{eq:intExpression}) can be evaluated either
numerically, or can be expressed as an infinite series involving
Meijer-G functions.  These results are confirmed in Fig.
\ref{fig:fig11}, where we compare the exact analytical capacity of
AF MIMO dual-hop systems, based on (\ref{eq:blah}) and
(\ref{eq:intExpression}), with Monte-Carlo simulated curves for two
different antenna and relay configurations.  In both cases, there is
exact agreement between the analysis and simulations, as expected.

\subsubsection{Analogies with Single-Hop MIMO Ergodic Capacity}
\label{sec:Analogy}

Let $C^{\rm SH-MIMO}(n_s, n_d, \rho)$ denote the ergodic capacity of
a conventional single-hop i.i.d.\ Rayleigh fading MIMO channel
matrix $\mathbf{H} \in \mathcal{C}^{n_d \times n_s}$, with $n_s$
transmit and $n_d$ receive antennas, and average SNR $\rho$; i.e.\
\begin{align}
C^{\rm SH-MIMO}(n_s, n_d, \rho) = E \left\{ \log_2 \det \left(
\mathbf{I}_{n_d} + \frac{\rho}{n_s} \mathbf{H} \mathbf{H}^\dagger
\right) \right\} \; .
\end{align}
Here, we demonstrate four particular cases for which the AF MIMO
dual-hop channel relates directly to single-hop MIMO channels, in
terms of ergodic capacity.



\begin{itemize}


\item
As the number of relay antennas grows large, i.e.\ $n_r \to \infty$,
the ergodic capacity of AF MIMO dual-hop systems becomes
\begin{align}\label{eq:exact_sp1}
\lim_{n_r \to \infty}  {C\left( \rho  \right)}  = \frac{1}{2} C^{\rm
SH-MIMO}\left(n_s, n_d, \frac{ \rho \alpha }{1 + \rho + \alpha}
\right) \; .
\end{align}

A proof is presented in Appendix \ref{sec:Proof_exactsp1}. Note that
a similar phenomenon has been derived in \cite{Bolcskei06}, for the
special case $n_s = n_d$. Here, (\ref{eq:exact_sp1}) generalizes
that result for arbitrary source and destination antenna
configurations.

\item As the number of source antennas grows large, i.e.\ $n_s \to
\infty$, the ergodic capacity of AF MIMO dual-hop systems becomes
\begin{align}\label{eq:exact_sp2}
\lim_{n_s \to \infty}  {C\left( \rho  \right)}  = \frac{1}{2} C^{\rm
SH-MIMO}\left(n_r, n_d, \alpha \right) - \frac{1}{2} C^{\rm
SH-MIMO}\left(n_r, n_d, \frac{ \alpha }{1 + \rho } \right) \; .
\end{align}
%

A proof is presented in Appendix \ref{sec:Proof_exactsp2}.
Interestingly, we see that as $\rho$ grows large, the right-most
term in (\ref{eq:exact_sp2}) disappears, and the AF MIMO dual-hop
capacity becomes equivalent to one half of the ergodic capacity of a
single-hop MIMO channel with $n_r$ transmit antennas, $n_d$ receive
antennas, and average SNR $\alpha$.

\item As the number of destination antennas grows large, i.e.\ $n_d \to
\infty$, the ergodic capacity of AF MIMO dual-hop systems becomes
\begin{align}\label{eq:exact_sp4}
\lim_{n_d \to \infty}  {C\left( \rho  \right)}  = \frac{1}{2} C^{\rm
SH-MIMO}\left(n_s, n_r, \rho
\right) \; .
\end{align}
The result is trivially obtained by directly taking $\lambda _i^2 \to
\infty$ in (\ref{eq:ergcapa}). We see that the AF MIMO dual-hop
capacity becomes equivalent to one half of the ergodic capacity of a
single-hop MIMO channel with $n_s$ transmit antennas, $n_r$ receive
antennas, and average SNR $\rho$.

\item As the power gain of the relay grows large, i.e.\ $\alpha \to
\infty$, the ergodic capacity of AF MIMO dual-hop systems becomes
\begin{align}\label{eq:exact_sp3}
\lim_{\alpha \to \infty}  {C\left( \rho  \right)}  = \frac{1}{2}
C^{\rm SH-MIMO}\left(n_s, q, \rho \right)  \; .
\end{align}
The result is trivially obtained by directly taking $\alpha \to
\infty$ in (\ref{eq:ergcapa}).
Thus we see the interesting result that even as the relay power gain
becomes very large, the capacity of AF MIMO dual-hop channels remains
bounded, and in fact becomes equivalent to one half of the ergodic
capacity of a single-hop MIMO channel with $n_s$ transmit antennas,
$q = \min(n_r, n_d)$ receive antennas, and average SNR $\rho$.

\end{itemize}

We note that for each of the cases (\ref{eq:exact_sp1})--(\ref{eq:exact_sp3}), closed-form expressions
can be obtained by directly invoking known results from the
single-hop MIMO capacity literature (eg.\ see \cite{Shin03}).



In order to obtain further simplified closed-form results, it is
useful to investigate the ergodic capacity in the high SNR regime.
This is presented in the subsection below.

\subsection{High SNR Capacity Analysis}\label{sec:lowcapaanalysis}

For the high SNR regime, we consider two important scenarios;
namely, one where the source and relay powers grow large
proportionately, and one where the source power grows large but the
relay power is kept fixed.

\subsubsection{Large Source Power, Large Relay Power}

Here we have $\alpha  \to \infty $, $\rho \to \infty$, with $\alpha
/\rho = \beta$, for some fixed $\beta$. Then $\rho a \to
\frac{\alpha }{{n_r }}$ and $a \to \beta /n_r $, and the ergodic
capacity at high SNR reduces to
\begin{align}\label{eq:ergodiccapabeta}
\left. {C\left( \rho  \right)} \right|_{\alpha ,\rho  \to \infty
,\alpha /\rho  = \beta }  = \frac{1}{2}E\left\{ {\log _2 \det \left(
{{\bf{I}}_{n_s }  + \frac{{\rho \beta }}{{n_s n_r }}{\bf{\tilde
H}}_1^\dag  {\bf{\bar L \tilde H}}_1 } \right)} \right\}
\end{align}
where ${\bf{\bar L}} = {\rm diag}\left\{ {\lambda _i^2 /\left( {1 +
\left( {\beta /n_r } \right)\lambda _i^2 } \right)} \right\}_{i =
1}^q $. We can express (\ref{eq:ergodiccapabeta}) in the general
form \cite{Lozano05}
\begin{align}\label{eq:highsnrcapa}
\left. {C\left( \rho  \right)} \right|_{\alpha ,\rho  \to \infty,
,\alpha /\rho  = \beta  } =
  S_\infty  \left( {\frac{{\left. {\rho }
\right|_{{\rm dB}} }}{{3{\rm dB}}} - \mathcal{L}_\infty  } \right) +
o\left( 1 \right)
\end{align}
where $ 3{\rm dB} = 10\log _{10} (2) $.  Here, the two key
parameters are $S_\infty$, which denotes the high-SNR slope in
bits/s/Hz/(3 ${\rm dB}$) given by
\begin{align}\label{eq:highsnrslope}
S_\infty   = \mathop {\lim }\limits_{\alpha,\rho  \to \infty }
\frac{\left. {C\left( \rho  \right)} \right|_{\alpha ,\rho  \to
\infty, \alpha /\rho = \beta  } }{{\log _2 (\rho) }}
\end{align}
and $\mathcal{L}_\infty$, which represents the high-SNR power offset
in 3 ${\rm dB}$ units given by
\begin{align}\label{eq:highsnroffset}
\mathcal{L}_\infty   = \mathop {\lim }\limits_{\alpha, \rho  \to
\infty } \left( {\log _2 (\rho)  - \frac{\left. {C\left( \rho
\right)} \right|_{\alpha ,\rho  \to \infty, \alpha /\rho = \beta  }
}{{S_\infty }}} \right) .
\end{align}
From (\ref{eq:ergodiccapabeta}), we can evaluate
$\mathcal{S}_\infty$ and $\mathcal{L}_\infty$ in closed-form as
follows.
\begin{theorem}\label{exact_sp3}
For the case $\alpha  \to \infty $, $\rho \to \infty$, with $\alpha
/\rho = \beta$, the high-SNR slope and high-SNR power offset of AF
MIMO dual-hop systems are given by
\begin{align}\label{eq:slope}
S_\infty = \frac{s}{2} {~~} {\rm bit/s/Hz/}(3 {\rm dB})
\end{align}
and\footnote{Note that here we explicitly indicate the dependence of
the high SNR power offset on $n_s$, $n_r$, and $n_d$.}
\begin{align}\label{eq:offset}
\mathcal{L}_\infty (n_s, n_r, n_d)   = \log _2 \left( {\frac{{n_s
n_r }}{\beta }} \right) - \frac{1}{{s\ln 2}}\left[ {\sum\limits_{k =
1}^s {\psi \left( {n_s  + k - s} \right)}  +
\mathcal{K}\sum\limits_{k = q - s + 1}^q {\det \left( {{\bf{\bar
W}}_k } \right)} } \right]
\end{align}
respectively, where ${{\bf{\bar W}}_k}$ is a $q \times q$ matrix
with entries
\begin{align}
\left\{ {{\bf{\bar W}}_k } \right\}_{m,n}  = \left\{
{\begin{array}{*{20}c}
   \left( {\frac{\beta }{{n_r }}} \right)^{1 - \tau
} \vartheta_{\tau-1}\left( \frac{\beta}{n_r} \right), & {n \ne k,}  \\
\varsigma_{m+n}\left( \frac{\beta}{n_r} \right), & {n = k.}  \\
%
\end{array}} \right. \;
\end{align}
For the case $q = s$ (i.e.\ corresponding to $\min(n_s, n_r, n_d) =
n_d$ or $\min(n_s, n_r, n_d) = n_r$), the high SNR power offset
(\ref{eq:offset}) admits the alternative form
\begin{align}\label{eq:poweroffset_sp0}
{\mathcal{L}_{\infty}  } (n_s, n_r, n_d)  &= \log _2 \left(
{\frac{{n_s n_r }}{\beta }} \right) - \frac{1}{{s\ln 2}}\left[
{\sum\limits_{k = 1}^s {\psi \left( {n_s  - s + k} \right)} +
\sum\limits_{i = 0}^{q - 1} {\sum\limits_{j = 0}^i {\sum\limits_{l =
0}^{2j} {\sum\limits_{k = 0}^{2q - l - 2} \binom{2q - l - 2}{k}
%
} } } } \right.\nonumber\\
& \hspace*{-1cm} \times\mathcal{A}\left( {i,j,l,p,q} \right)\left. {
\left( {\frac{\beta }{{n_r }}} \right)^{2q - l - 2 - k} \Gamma
\left( {p + q - k - 1} \right)\left( {\psi \left( {p + q - k - 1}
\right) - \sum\limits_{m = 0}^{p + q - k - 2} {g_m \left(
{\frac{{n_r }}{\beta }} \right)} } \right)} \right] \; .
\end{align}
\end{theorem}
\begin{proof}
See Appendix \ref{sec:Proof_exactsp3}.
\end{proof}
Interestingly, we see that the high SNR slope depends only on the
minimum system dimension, i.e.\ $s = \min (n_s, n_r, n_d)$, whereas
the high SNR power offset is a much more intricate function of
$n_s$, $n_r$, and $n_d$.  Fig. \ref{fig:fig11} depicts the
analytical high SNR capacity approximations for AF MIMO dual-hop
systems, based on (\ref{eq:slope}) and (\ref{eq:offset}). These
approximations are seen to converge to their respective exact
capacity curves for quite moderate SNR levels (eg. $< 20 \rm dB$).

It is important to note that \emph{Theorem \ref{exact_sp3}} presents
an exact characterization of the key high SNR ergodic capacity
parameters, $\mathcal{S}_\infty$ and $\mathcal{L}_\infty(\cdot)$,
for arbitrary numbers of antennas at the source, relay, and
destination terminals.  We now examine some particularizations of
\emph{Theorem \ref{exact_sp3}}, in which these expressions reduce to
simple forms.

%


\begin{corollary}
Let $n_r = 1$.  Then $\mathcal{S}_\infty = 1/2$, and
$\mathcal{L}_\infty (\cdot)$ reduces to
\begin{align}
 {\mathcal{L}_{\infty}  } (n_s, 1, n_d)  = \log _2 \left(
{\frac{{n_s }}{\beta }} \right) - \frac{1}{{\ln 2}}\left[ {\psi
\left( {n_s } \right) + \psi \left( n_d \right) - \sum\limits_{m =
0}^{n_d - 1} {g_m \left( {\frac{{1 }}{\beta }} \right)} } \right] \;
.
\end{align}
\end{corollary}
Note that, as $n_s$ grows large, $\psi \left( {n_s } \right) = \ln
n_s + o(1)$ \cite[Eq. 6.3.18.]{Abramowitz74}, where the $o(1)$ term
disappears as $n_s \to \infty$, and as such we have
\begin{align}
\lim_{n_s \to \infty} {\mathcal{L}_{\infty}  } (n_s, 1, n_d) = \log
_2 \left( {\frac{1}{\beta }} \right) - \frac{1}{{\ln 2}}\left[ {\psi
\left( n_d \right) - \sum\limits_{m = 0}^{n_d - 1} {g_m \left(
{\frac{{1}}{\beta }} \right)} } \right] \; .
\end{align}

\begin{corollary}
Let $n_d = 1$.  Then $\mathcal{S}_\infty = 1/2$, and
$\mathcal{L}_\infty (\cdot)$ reduces to
\begin{align}
 {\mathcal{L}_{\infty}  } (n_s, n_r, 1)  = \log _2 \left(
{\frac{{n_s n_r }}{\beta }} \right) - \frac{1}{{\ln 2}}\left[ {\psi
\left( {n_s } \right) + \psi \left( n_r \right) - \sum\limits_{m =
0}^{n_r - 1} {g_m \left( {\frac{{n_r }}{\beta }} \right)} } \right]
\; .
\end{align}
In this case, as $n_s$ grows large we have
\begin{align}
\lim_{n_s \to \infty} {\mathcal{L}_{\infty}  } (n_s, n_r, 1) = \log
_2 \left( {\frac{n_r}{\beta }} \right) - \frac{1}{{\ln 2}}\left[
{\psi \left( n_r \right) - \sum\limits_{m = 0}^{n_r - 1} {g_m \left(
{\frac{{n_r }}{\beta }} \right)} } \right] \; .
\end{align}
\end{corollary}
Based on these results, we can easily examine the effect of the
relative power gain factor $\beta$ on the ergodic capacity.  In
particular, noting that $g_l \left( x \right)$ in (\ref{eq:gDefn})
is a monotonically decreasing function of $x$ in the
interval\footnote{This conclusion is easily established by noting
that ${\rm d}/{\rm d}x \left( g_l \left( x \right) \right) = e^x
\left[ {E_{l + 1} \left( x \right) - E_l \left( x \right)} \right]
$, and using \cite[Eq. 5.1.17]{Abramowitz74}.} $\left[ {0,\infty }
\right)$, we see that increasing $\beta$, whilst having no effect on
the high SNR capacity slope $\mathcal{S}_\infty$, results in
decreasing the high SNR power offset $\mathcal{L}_\infty(\cdot)$,
and therefore increasing the ergodic capacity in the high SNR
regime.


\begin{corollary}\label{poweroffset_sp1}
Let $n_s = n_r = 1$. Adding $k$ destination antennas, while not
altering $S_\infty$, would reduce the high SNR power offset as
\begin{align}\label{eq:poweroffset_sp1}
\delta (n_d, k) &\defeq \mathcal{L}_\infty  \left( {1,1,n_d + k}
\right) -
\mathcal{L}_\infty \left( {1,1,n_d} \right) \nonumber \\
 &=
- \frac{1}{{\ln 2}} { \sum\limits_{l = n_d}^{n_d + k - 1} \left(
\frac{1}{\ell} + {g_l \left( {\frac{1}{\beta }} \right)} \right) }
\; .
\end{align}
\end{corollary}
Note that, to obtain this result, we have invoked the definition of
the digamma function \cite{Gradshteyn00}.
Since $ g_l \left( x \right) > 0$ for $x \in \left[ {0,\infty }
\right)$, it is clear that the high SNR power offset
$\mathcal{L}_\infty (\cdot)$ in (\ref{eq:poweroffset_sp1}) is a
decreasing function of $k$, thereby confirming the intuitive notion that adding more antennas to the
destination terminal has the effect of improving the ergodic
capacity.

\begin{example}
With respect to $\beta = 1$,
\begin{align}
& \mathcal{L}_\infty  \left( {1,1,2} \right) = \mathcal{L}_\infty
\left( {1,1,1} \right) - 2.58 \; {\rm dB} \\
& \mathcal{L}_\infty  \left( {1,1,3} \right) = \mathcal{L}_\infty
\left( {1,1,1} \right) - 3.46 \; {\rm dB} \\
& \mathcal{L}_\infty  \left( {1,1,\infty} \right) =
\mathcal{L}_\infty \left( {1,1,1} \right) - 5.08 \; {\rm dB}
\end{align}
where $\mathcal{L}_\infty \left( {1,1,1} \right) =$ { 7.57} dB.
\end{example}


%

\begin{table} \label{table_poweroffset1}
\centering \caption{High SNR offset as function of $n_d$, where $n_s
= 2$, $n_r = 3$ and $\beta = 2$}
\begin{tabular}{|c|c|c|c|c|c|c|}
\hline
 $n_d$ & 4 & 6 & 8 & 10 & 12 & 14 \\    
\hline
$\mathcal{L}_\infty$ (dB) & 2.593 & 1.573 & 1.147 & 0.88 & 0.73 & 0.622 \\ 
\hline
\end{tabular}
\end{table}

\begin{table} \label{table_poweroffset2}
\centering \caption{High SNR offset as function of $n_r$, where $n_s
= 2$, $n_d = 4$ and $\beta = 2$}
\begin{tabular}{|c|c|c|c|c|c|c|}
\hline
 $n_r$ & 3 & 5 & 7 & 9 & 11 & 13\\
\hline
$\mathcal{L}_\infty$ (dB) & 2.593 & 1.251 & 0.847 & 0.636  & 0.493 & 0.429\\
\hline
\end{tabular}
\end{table}

\begin{figure}
\centering
\includegraphics[scale=0.8]{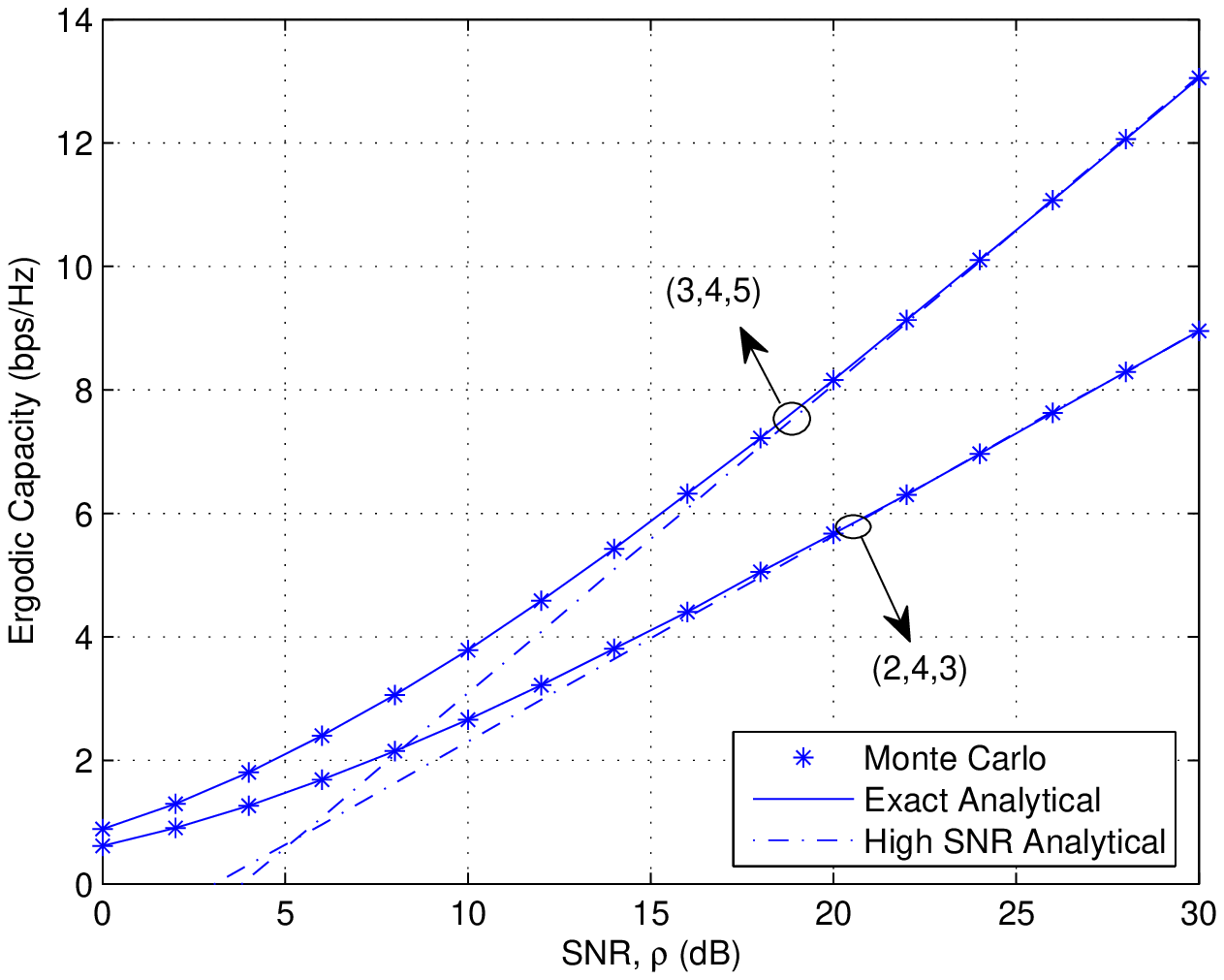}
\captionstyle{mystyle2}\caption{Comparison of exact
analytical, high SNR analytical, and Monte Carlo simulation results for ergodic capacity of
AF MIMO dual-hop systems  with different antenna configurations.  Results are shown for
$\alpha / \rho = 2$.}\small
\label{fig:fig11}
\end{figure}

\begin{figure}
\centering
\includegraphics[scale=0.8]{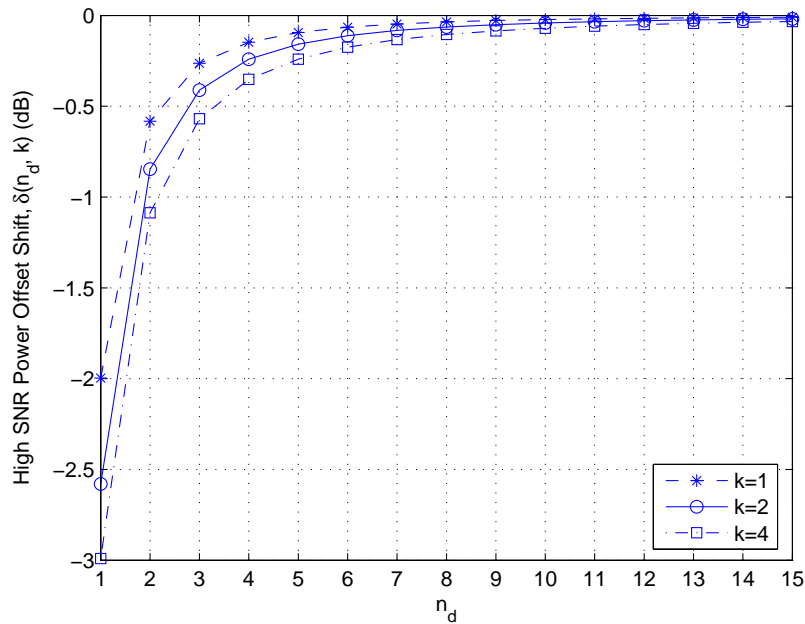}
\captionstyle{mystyle2}\caption{High SNR power offset shift, in decibels, obtaining by adding either (a) one antenna to the destination, (b) two antennas to the destination, or (c) four antennas to the destination.  Results are shown for $n_s = n_r = 1$ and $\alpha / \rho = 2$.}
\label{fig:fig13}
\end{figure}

Fig. \ref{fig:fig13} illustrates the relationship in
\textit{Corollary \ref{poweroffset_sp1}}, where the high SNR power
offset shift $\delta(n_d, k)$ is plotted against $n_d$, for $k=1$,
$k=2$, and $k = 4$. As expected, for a fixed value of $k$,
$\delta(n_d, k)$ is an increasing function of $n_d$, approaching a
limit of $0$ dB as $n_d \to \infty$. Table I and Table II present
the high SNR power offset as a function of $n_d$ and $n_r$
respectively, for $n_s = 2$. We see that when $n_d$ (resp.\ $n_r$)
is small, then a small increase in $n_d$ (resp.\ $n_r$) yields a
significant improvement in terms of the high SNR power offset.
However, in agreement with Fig. \ref{fig:fig13}, adding more and
more antennas yields diminishing returns.


\subsubsection{Large Source Power, Fixed Relay Power}
\label{sec:lowcapaanalysis2} Here we take $\rho \to \infty$ and keep
$\alpha$ fixed. Then, noting that $\left. {\rho a} \right|_{\rho
\to \infty } \to \alpha/n_r$, the ergodic capacity reduces to
\begin{align}\label{eq:ergodicrholarge}
\lim_{\rho \to \infty}  {C\left( \rho  \right)}   =
\frac{s}{2}E\left\{ {\log _2 \left( {1 + \frac{\alpha }{{n_s n_r
}}\tilde \lambda } \right)} \right\}
\end{align}
where ${\tilde \lambda }$ denotes an unordered eigenvalue of
${{\bf{\tilde H}}_1^\dag  {\bf{\tilde L\tilde H}}_1 }$. Using
\textit{Corollary \ref{final_unorderedproduct}}, we can evaluate
this constant as
\begin{align} \label{eq:HighSNRExpr}
\lim_{\rho \to \infty}  {C\left( \rho  \right)}   =
\frac{\mathcal{K}}{{\ln 2}}\sum\limits_{l = 1}^q\sum\limits_{k = q -
s + 1}^q { { {{{\bar G}}_{l,k} } } } {\mathcal {F}}_{l,k}
\end{align}
where
\begin{align}
{\mathcal {F}}_{l,k} = \int_0^\infty  {\ln \left( {1 + \frac{\alpha
}{{n_s n_r }}y} \right)} y^{\left( {n_s  + 2k + p + l - 2q - 3}
\right)/2} K_{p + l - n_s  - 1} \left( {2\sqrt y } \right)dy.
\end{align}
To evaluate the remaining integral in (\ref{eq:ergodicrholarge}), we
first express the logarithm in terms of the Meijer G-function as
\cite[Eq. 8.4.6.5]{Prudnikov90}
\begin{align}\label{eq:logMeijerG}
{\log _2 \left( {1 + \frac{\alpha }{{n_s n_r }}\tilde \lambda } \right)}
= \frac{1}{{\ln 2}}G_{2,2}^{1,2} \left( {\frac{\alpha }{{n_s n_r
}}\tilde \lambda \left| {\begin{array}{*{20}c}
   {1,} & 1  \\
   {1,} & 0  \\
\end{array}} \right.} \right)  \;
\end{align}
and then apply the integral relationships \cite[Eq.\
7.821.3]{Gradshteyn00} and \cite[Eq.\ 9.31.1]{Gradshteyn00}. This
leads to the following closed-form expression for the ergodic
capacity of AF MIMO dual-hop systems as the source power $\rho$
grows large for fixed relay power $\alpha$,
\begin{align} \label{eq:CHighSNRFixedAlpha}
& \lim_{\rho \to \infty} {C\left( \rho  \right)} =
\frac{\mathcal{K}}{{2\ln 2}}\sum\limits_{l = 1}^q\sum\limits_{k = q
- s + 1}^q { {\frac{{{{{\bar G}}_{l,k} } }}{{\Gamma \left( {n_s - q
+ k} \right)}}} } \nonumber \\
& \hspace*{3cm} \times G_{2,4}^{4,1} \left( {\frac{{n_s n_r
}}{\alpha }\left| {\begin{array}{*{20}c}
   {0,} & {1,} & {} & {}  \\
   {k + p + l - q - 1,} & {n_s  + k - q,} & {0,} & 0  \\
\end{array}} \right.} \right) \; .
\end{align}
This result shows that if we fix $\alpha$ and take $\rho$ large,
then the ergodic capacity of AF MIMO dual-hop systems remains
bounded (as a function of $\alpha$). This confirms the intuitive
notion that the capacity is restricted by the weakest link in the
relay network; in this case, the relay-destination link.

\section{Tight Bounds on the Ergodic Capacity}\label{sec:applicationbounds}

In order to obtain further simplified closed-form results, in this
section we derive new upper and lower bounds on the ergodic
capacity.

\subsection{Upper Bound}

The following theorem presents a new tight upper bound on the ergodic capacity of AF
MIMO dual-hop systems.
\begin{theorem}\label{upperbound}
The ergodic capacity of AF MIMO dual-hop systems is upper bounded by
\begin{align}\label{eq:upperbound}
C\left( \rho  \right) \le C_U (\rho)  =  \frac{1}{2}\log _2 \left(
{\mathcal{K}\det ( {\bf{\bar \Xi }} )} \right)
\end{align}
where ${\bf{\bar \Xi }}$ is defined in (\ref{eq:XiBar}).
\end{theorem}
\begin{proof}
Application of Jensen's inequality gives\footnote{Note that this inequality has also been applied
in the ergodic capacity analysis of single-user single-hop MIMO
systems (see eg.\ \cite{Zhang05,Mckay05,Jin07}).}
\begin{align}
C\left( \rho  \right) \leqslant \frac{1} {2}\log _2 E\left\{ {\det
\left( {{\mathbf{I}}_{n_s }  + \frac{{\rho a}}{{n_s
}}{\mathbf{\tilde H}}_1^\dag  {\mathbf{L\tilde H}}_1 } \right)}
\right\} \; .
\end{align}
The result now follows by using \textit{Theorem \ref{expdetNew}}.
\end{proof}

\begin{figure}
\centering
\includegraphics[scale=0.8]{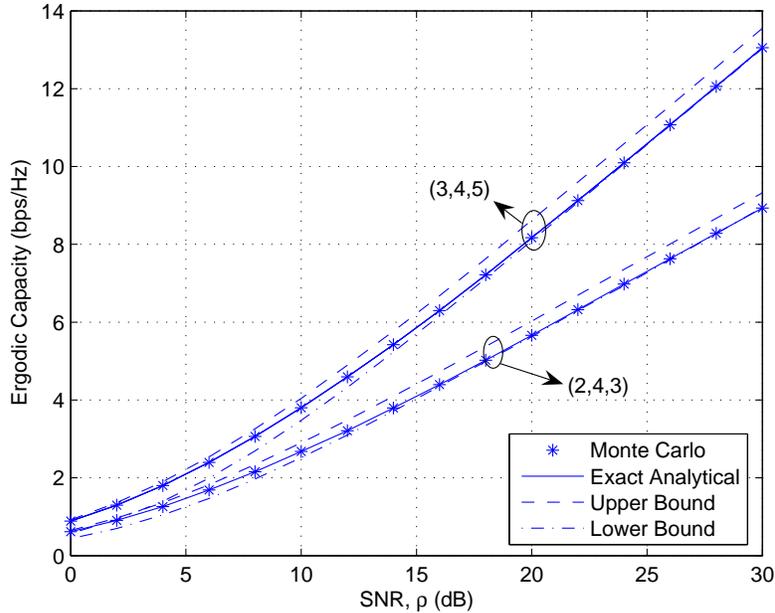}
\captionstyle{mystyle2}\caption{Comparison of bounds, exact
analytical, high SNR analytical, and Monte Carlo simulation results
for ergodic capacity of AF MIMO dual-hop systems  with different
antenna configurations. Results are shown for $\alpha / \rho = 2$.} \label{fig:fig10}
\end{figure}

Fig. \ref{fig:fig10} compares the closed-form upper bound
(\ref{eq:upperbound}) with the exact analytical ergodic capacity
based on (\ref{eq:blah}) and (\ref{eq:intExpression}), for two
different
AF MIMO dual-hop system configurations.
The results are shown as a function of SNR $\rho$, with $\alpha = 2\rho$.
We see that the closed-form upper bound is very tight for all SNRs,
for both system configurations considered.  Moreover, we see that in
the low SNR regime (e.g.\ $\rho \approx 5$ dB), the upper bound and
exact capacity curves coincide.

The ensuing corollaries present some example scenarios for which the
upper bound (\ref{eq:upperbound}) reduces to simplified forms.

%



\begin{corollary}\label{upperbound_sp3}
For the case $n_s  \to \infty $, $C_U (\rho)$ becomes
\begin{align} \label{eq:upperbound_sp3}
\lim_{n_s  \to \infty} C_U (\rho)  = \frac{1}{2}\log _2 \left(
{\mathcal{K}{ \det ( {\bf{\bar \Xi }}_1 )}} \right)
\end{align}
where ${{\bf{\bar \Xi }}}_1$ is a $q \times q$ matrix with entries
\begin{align}
\left\{ {{\bf{\bar \Xi }}}_1 \right\}_{m,n}  =   a^{1 - \tau } \vartheta_{\tau-1}( a )  + \rho a^{1 - \tau }  \vartheta_{\tau}( a ) \; .
\end{align}
\end{corollary}
\begin{proof}
The proof is straightforward and is omitted.
\end{proof}
This result shows that in AF MIMO dual-hop systems, when the numbers
of antennas at both the relay and destination remain fixed, the
ergodic capacity remains bounded as the number of source antennas
grows large. This is in agreement with the results in Section
\ref{sec:Analogy}.

Note that for the scenarios $n_r \to \infty$ and $n_d \to \infty$,
simplified closed-form results can also be obtained by taking the
corresponding limits in (\ref{eq:upperbound_sp3}) or, alternatively,
by using the equivalent single-hop MIMO capacity relations in
(\ref{eq:exact_sp1}) and (\ref{eq:exact_sp4}), and applying known
upper bounds for single-hop MIMO channels in \cite{Oyman03}.  We
omit these expressions here for the sake of brevity.


\begin{corollary}\label{upperbound_sp4}
Let $n_r = 1$.  Then, $C_U (\rho)$ reduces to
\begin{align}\label{eq:upboundsp4}
C_U^{n_r  = 1}(\rho) = \frac{1}{2}\log _2 \left( {1 + \rho n_d e^{\frac{{1 + \rho }}{\alpha }} E_{n_d  + 1} \left( {\frac{{1 + \rho }}{\alpha }} \right)} \right).
\end{align}
When $n_d \to \infty$, $C_U^{n_r  = 1}(\rho)$ becomes
\begin{align}\label{eq:Cusp4}
\lim_{n_d \to \infty} C_U^{n_r  = 1} (\rho) = \frac{1}{2}\log _2
\left( {1 + \rho } \right).
\end{align}
When $\alpha \to \infty$, $C_U^{n_r  = 1}(\rho)$ becomes
\begin{align}\label{eq:Cusp41}
\lim_{\alpha  \to \infty} C_U^{n_r  = 1}(\rho)  =
\frac{1}{2}\log _2 \left( {1 + \rho } \right) \; .
\end{align}
\end{corollary}
\begin{proof}
See Appendix \ref{sec:Proof_Upperboundsp4}.
\end{proof}
This shows the interesting result that, if a single relay antenna is
employed, then when either the number of destination antennas $n_d$
or the relay gain $\alpha$ grows large, the ergodic capacity is upper
bounded by the capacity of an AWGN SISO channel.

\begin{corollary}\label{upperbound_sp1}
In the high SNR regime, (i.e.\ as $\rho \to \infty$) for fixed relay
gain $\alpha$, $C_U (\rho)$ becomes
\begin{align}\label{eq:upperboundsp1}
\lim_{\rho  \to \infty}  C_U(\rho) = \frac{1}{2}\log _2 \left(
{\mathcal{K}\det ( {{\bf{\tilde \Xi }}} )} \right)
\end{align}
where ${{\bf{\tilde \Xi }}}$ is a $q \times q$ matrix with entries
\begin{align}
\left\{ {{\bf{\tilde \Xi }}} \right\}_{m,n}  = \left\{
{\begin{array}{*{20}c}
   {\Gamma \left( {\tau  - 1} \right)}, & {n \le q - n_s, }  \\
   {\Gamma \left( {\tau  - 1} \right)\left( {1 + \frac{\alpha }{{n_s n_r }}\left( {n_s  - q + n} \right)\left( {\tau  - 1} \right)} \right)}, & {n > q - n_s. }  \\
\end{array}} \right. \;
\end{align}
\end{corollary}
\begin{proof}
See Appendix \ref{sec:Proof_Upperboundsp1}.
\end{proof}
This expression is clearly much simpler than the exact ergodic
capacity expression given for this regime in
(\ref{eq:CHighSNRFixedAlpha}).


\subsection{Lower Bound}

The following theorem presents a new tight lower bound on the ergodic capacity of AF
MIMO dual-hop systems.

\begin{theorem}\label{lowerbound}
The ergodic capacity of AF MIMO dual-hop systems is lower bounded by
\begin{align}\label{eq:lowerbound}
C\left( \rho  \right) &\ge C_L (\rho) = \frac{s}{2}\log _2 \left( {1
+ \frac{{\rho a}}{{n_s }}\exp \left( \frac{1}{s}\left[
{\sum\limits_{k = 1}^s {\psi \left( {n_s  - s + k} \right)} +
\mathcal{K}\sum\limits_{k = q - s + 1}^q {\det \left( {{\bf{ W}}_k }
\right)} } \right]
 \right)} \right)
\end{align}
where $\mathbf{W}_k$ is defined as in (\ref{eq:WkDefn}).
\end{theorem}
\begin{proof}
See Appendix \ref{sec:Proof_Lowerbound}.
\end{proof}
In Fig. \ref{fig:fig10}, this closed-form lower bound is compared
with the exact ergodic capacity of AF MIMO dual-hop systems. Results
are shown for different system configurations. The lower bound is
clearly seen to be tight for the entire range of SNRs. Moreover, in
the high SNR regime (e.g.\ $\rho \approx 15$ dB), we see that the
lower bound and exact capacity curves coincide.


The ensuing corollaries present some example scenarios for which the
lower bound (\ref{eq:lowerbound}) reduces to simplified forms.

\begin{corollary}\label{lowerbound_sp3}
For the case $n_s  \to \infty $, $C_L (\rho)$ reduces to
\begin{align}
\lim_{{n_s  \to \infty }} C_L \left( \rho  \right)  =
\frac{s}{2}\log _2 \left( {1 + \rho a\exp \left(
{\frac{\mathcal{K}}{s}\sum\limits_{k = 1}^q {\det \left(\mathbf{W}_k
 \right)} } \right)} \right) \; .
\end{align}
\end{corollary}
\begin{proof}
See Appendix \ref{sec:Proof_Lowerboundsp3}.
\end{proof}
Again, we note that for the scenarios $n_r \to \infty$ and $n_d \to
\infty$, simplified closed-form results can also be obtained by
taking the corresponding limits in (\ref{eq:upperbound_sp3}) or,
alternatively, by using (\ref{eq:exact_sp1}) and
(\ref{eq:exact_sp4}), and applying known lower bounds for single-hop
MIMO channels in \cite{Oyman03}.




\begin{corollary}\label{lowerbound_sp4}
For the case $n_r = 1$, $C_L (\rho)$ reduces to
\begin{align}\label{eq:lowerboundsp4}
C_L^{n_r  = 1} \left( \rho  \right) = \frac{1}{2}\log _2 \left( {1 +
\frac{{\rho \alpha }}{{n_s \left( {1 + \rho } \right)}}\exp \left(
{\psi \left( {n_s } \right) + \psi \left( n_d \right) - e^{\left( {1
+ \rho } \right)/\alpha } \sum\limits_{l = 0}^{n_d - 1} {E_{l + 1}
\left( {\frac{{1 + \rho }}{\alpha }} \right)} } \right)} \right).
\end{align}
When $n_s \to \infty$, $C_L^{n_r  = 1}(\rho)$ becomes
\begin{align}\label{eq:lsp4ns}
\mathop {\lim }\limits_{n_s  \to \infty } C_L^{n_r  = 1} \left( \rho
\right) = \frac{1}{2}\log _2 \left( {1 + \frac{{\rho \alpha }}{{
 {1 + \rho } }}\exp \left( {\psi \left( {n_d } \right)
- e^{\left( {1 + \rho } \right)/\alpha } \sum\limits_{l = 0}^{n_d  -
1} {E_{l + 1} \left( {\frac{{1 + \rho }}{\alpha }} \right)} }
\right)} \right) \; .
\end{align}
When $n_d \to \infty$,  $C_L^{n_r  = 1}(\rho)$ becomes
\begin{align}\label{eq:lsp4nd}
\mathop {\lim }\limits_{n_d  \to \infty } C_L^{n_r  = 1} \left( \rho
\right) = \frac{1}{2}\log _2 \left( {1 + \frac{{\rho \alpha }}{{n_s \left( {1 + \rho } \right)}}\exp \left( {\psi \left( {n_s } \right) + \psi \left( {\frac{{1 + \rho }}{\alpha }} \right)} \right)} \right) \; .
\end{align}
When $\alpha \to \infty$, $C_L (\rho)$ becomes
\begin{align}\label{eq:Clsp41}
\lim_{\alpha  \to \infty} C_L^{n_r  = 1} \left( \rho  \right) =
\frac{1}{2}\log _2 \left( {1 + \frac{\rho }{{n_s }}\exp \left( {\psi
\left( {n_s } \right)} \right)} \right) \; .
\end{align}
\end{corollary}
\begin{proof}
See Appendix \ref{sec:Proof_Lowerboundsp4}.
\end{proof}
As also observed from the upper bound in \textit{Corollary
\ref{upperbound_sp4}}, this result shows that for a system with a
single relay antenna, when the relay gain $\alpha$ grows large, the
ergodic capacity of an AF MIMO dual-hop channel is lower bounded by
the capacity of an AWGN SISO channel (with scaled average SNR).


\begin{figure}
\centering
\includegraphics[scale=0.8]{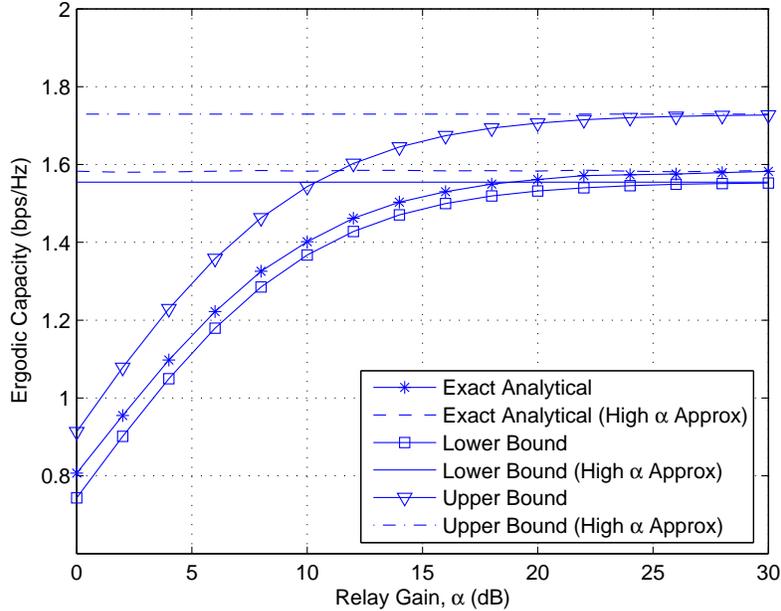}
\captionstyle{mystyle2}\caption{Comparison of capacity bounds, high $\alpha$
approximation, and exact analytical results for different relay gains. Results are shown for $n_r = 1$, $n_s
= 2$, $n_d = 4$ and $\rho = 10 {\rm dB}$.}
\label{fig:fig7}
\end{figure}

Fig. \ref{fig:fig7} plots the closed-form upper bound
(\ref{eq:upboundsp4}), closed-form lower bound
(\ref{eq:lowerboundsp4}), and the exact analytical ergodic capacity
based on (\ref{eq:blah}) and (\ref{eq:intExpression}), for an AF MIMO
dual-hop system with $n_r = 1$.  The results are presented as a
function of the relay gain $\alpha$. We see that both the upper and
lower bounds are quite tight for the entire range of $\alpha$
considered. The asymptotic approximations for the upper and lower
bounds, based on (\ref{eq:Cusp41}) and
(\ref{eq:Clsp41}) respectively, are also shown for further
comparison, and are seen to converge for moderate values of $\alpha$
(e.g.\ within $\alpha \approx 20$ dB).


\begin{figure}
\centering
\includegraphics[scale=0.8]{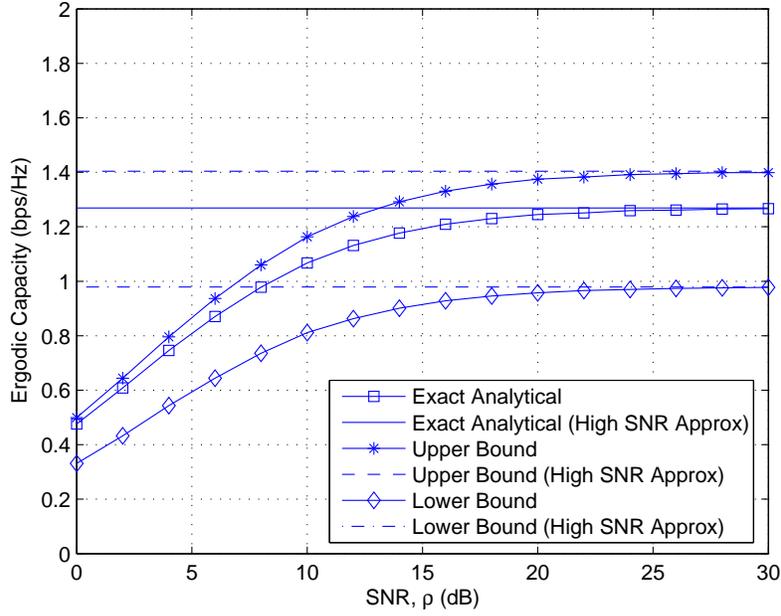}
\captionstyle{mystyle2}\caption{Comparison of capacity bounds, high SNR
approximations, and exact analytical results. Results are shown for a system configuration $(3, 4, 2)$ and
$\alpha = 2$.}
\label{fig:fig1}
\end{figure}

\begin{corollary}\label{lowerbound_sp1}
In the high SNR regime, (i.e.\ as $\rho \to \infty$) for fixed relay
gain $\alpha$, $C_L (\rho)$ becomes
\begin{align}\label{eq:lowerboundsp1}
\lim_{\rho  \to \infty } C_L (\rho) =  \frac{s}{2}\log _2 \left( {1
+ \frac{\alpha }{{n_r n_s }}\exp \left( {\frac{\mathcal{K}}{s}
 \sum\limits_{k = q - s + 1}^q {\det \left( {{\bf{\tilde W}}_k } \right)} } \right)} \right),
\end{align}
where ${\bf{\tilde W}}_k$ is a $q \times q$ matrix with entries
\begin{align}
\left\{ {{\bf{\tilde W}}_k } \right\}_{m,n}  = \left\{
{\begin{array}{*{20}c}
   {\Gamma \left( {\tau  - 1} \right)}, & {n \ne k}  \\
   {\Gamma \left( {\tau  - 1} \right)\left[ {\psi \left( {n_s  - q + n} \right) + \psi \left( {\tau  - 1} \right)} \right]}, & {n = k}  \\
\end{array}} \right. \; .
\end{align}
\end{corollary}
\begin{proof}
See Appendix \ref{sec:Proof_Lowerboundsp1}.
\end{proof}
As for the high SNR upper bound presented in
(\ref{eq:upperboundsp1}), this closed-form lower bound expression is
simpler than the exact ergodic capacity expression given for this
regime in (\ref{eq:CHighSNRFixedAlpha}).

Fig. \ref{fig:fig1} depicts the closed-form high SNR approximations
for the exact ergodic capacity, as well as the respective upper and
lower bounds, based on (\ref{eq:logMeijerG}),
(\ref{eq:upperboundsp1}), and (\ref{eq:lowerboundsp1}) respectively.
For comparison, curves are also presented for the upper bound
(\ref{eq:upperbound}), lower bound (\ref{eq:lowerbound}), and the
exact analytical ergodic capacity based in (\ref{eq:blah}) and
(\ref{eq:intExpression}).  Results are shown for an AF MIMO dual-hop
system with configuration $(3, 4, 2)$.  Clearly, the analytical high
SNR approximations are seen to be very accurate for even moderate
SNR levels (e.g.\ $\rho \approx 20$ dB).

\section{Conclusions}\label{sec:conclusion}

This paper has presented an analytical characterization of the ergodic capacity of AF MIMO dual-hop relay
channels under the common assumption that CSI is available at
the destination terminal, but not at the relay or the source terminal.  We derived a new exact expression for the ergodic capacity, as well as simplified and insightful closed-form expressions for the high SNR regime.  Simplified closed-form upper and lower bounds were also presented, which were shown to be tight for all SNRs.  The analytical results were made possible by first employing random
matrix theory techniques to derive new expressions for the p.d.f. of
an unordered eigenvalue, as well as random determinant results for the equivalent AF MIMO dual-hop relay channel, described by a certain product of finite-dimensional complex random matrices. The analytical results were validated through comparison with numerical simulations.

\section{Acknowledgement}

The authors would like to thank Dr. Veniamin I.\ Morgenshtern for providing the source code used to generate the asymptotic eigenvalue distributions in Fig. \ref{fig:fig12}.


\appendices

\section{Proofs of New Random Matrix Theory Results}

\subsection{Proof of Lemma \ref{unorderedpdf}}\label{sec:Proof_unorderedcond}

To prove this lemma, it is convenient to give a separate treatment for
the two cases, $q < n_s$ and $ q \ge n_s $.

\subsubsection{The $q < n_s$ Case}\label{sec:proofTh1part1}

For this case, an expression for the p.d.f.\ $f_{\lambda |
\mathbf{L}} ( \cdot )$ has been given previously as \cite{Alfano04}
\begin{align} \label{eq:unorderedPDF1}
f_{\lambda | \mathbf{L}} ( \lambda ) = \frac{\sum\limits_{l = 1}^q
{\sum\limits_{k = 1}^q { \lambda ^{n_s  - q + k - 1} e^{ - \lambda
/\beta _l } {\tilde D}_{l,k}} } }{ q\det \left( {\bf{L }}
\right)^{n_s - q + 1} \prod\nolimits_{i = 1}^{q} {\Gamma \left( {n_s
- i + 1} \right)} \prod_{i<j}^q ( \beta_j - \beta_i )}
\end{align}
where ${\tilde D}_{l,k}$ is the $(l,k)$th cofactor of a $q \times q$
matrix with entries
\begin{align}
\bigl\{ {\bf{\tilde D}} \bigr\}_{i,j}  = \Gamma \left( {n_s - q + j}
\right)\beta _i^{n_s  - q + j} \; .
\end{align}
After some basic manipulations, we can express this cofactor as
\begin{align}\label{eq:dexp}
{\tilde D}_{l,k} = \frac{{\prod\nolimits_{j = 1}^{q} {\Gamma \left(
{n_s - j + 1} \right)} }}{{\Gamma \left( {n_s  - q + k}
\right)}}\frac{{\det \left( {\bf{L }} \right)^{n_s - q + 1}
}}{{\beta _l^{n_s - q + 1} }} D_{l,k} \; .
\end{align}
Substituting (\ref{eq:dexp}) into (\ref{eq:unorderedPDF1}) yields
the desired result.

\subsubsection{The $q \ge n_s $ Case}\label{sec:proofTh1part2}
%

For this case, we start by employing a result from \cite[Eq.
11]{Smith03} to express the joint p.d.f. of the unordered
eigenvalues $ \gamma _1 ,  \ldots , \gamma _{n_s}$ of
${\mathbf{\tilde H}}_1^\dag {\mathbf{L\tilde H}}_1 $, conditioned on
$\mathbf{L}$, as follows
\begin{align}\label{eq:unorderedPDF}
f\left( {\left. {\gamma _1 , \ldots ,\gamma _{n_s} } \right|{\bf{L
}}} \right) = \frac{{\det \left( {{\bf{\Delta }}_1 }
\right)\prod\nolimits_{i < j}^{n_s} {\left( {\gamma _j  - \gamma _i
} \right)} }}{n_s{\prod\nolimits_{i = 1}^{n_s} { \Gamma \left( {n_s - i
+ 1} \right)} \prod\nolimits_{i < j}^{q} {\left( {\beta _j  - \beta
_i } \right)} }},
\end{align}
where ${{\bf{\Delta }}_1 }$ is the $q \times q$ matrix
\begin{align} \label{eq:Delta1Defn}
{\bf{\Delta }}_1  = \left[ {\begin{array}{*{20}c}
   1 & {\beta _1 } &  \cdots  & {\beta _1^{q - n_s - 1} } & {\beta _1^{q - n_s - 1} e^{ - \frac{{\gamma _1 }}{{\beta _1 }}} } &  \cdots  & {\beta _1^{q - n_s - 1} e^{ - \frac{{\gamma _{n_s} }}{{\beta _1 }}} }  \\
    \vdots  &  \vdots  &  \ddots  &  \vdots  &  \vdots  &  \ddots  &  \vdots   \\
   1 & {\beta _{q} } &  \cdots  & {\beta _{q}^{q - n_s - 1} } & {\beta _{q}^{q - n_s - 1} e^{ - \frac{{\gamma _1 }}{{\beta _{q} }}} } &  \cdots  & {\beta _{q}^{q - n_s - 1} e^{ - \frac{{\gamma _{n_s} }}{{\beta _{q} }}} }  \\
\end{array}} \right] \; .
\end{align}
The p.d.f. of a single unordered eigenvalue $\lambda$ is found from
(\ref{eq:unorderedPDF}) via
\begin{align}\label{eq:unorderedPDF_New}
f_{\lambda|\mathbf{L}}(\lambda) &=  \int_{0}^{\infty} \cdots
\int_{0}^{\infty}  f\left( {\left. {\gamma _1 , \ldots ,\gamma
_{n_s} } \right|{\bf{L }}} \right)  d \gamma_1 \cdots d
\gamma_{n_s-1} \biggr|_{\gamma_{n_s} = \lambda} \nonumber \\
& = \frac{1}{ {{n_s\prod\nolimits_{i = 1}^{n_s} { \Gamma \left( {n_s -
i + 1} \right)} \prod\nolimits_{i < j}^{q} {\left( {\beta _j  -
\beta _i } \right)} }} } \int_{0}^{\infty} \cdots \int_{0}^{\infty}
{ \det \left( {{\bf{\Delta }}_1 } \right) \det \left( {\gamma _i^{j
- 1} } \right) }  d \gamma_1 \cdots d \gamma_{n_s-1}
\biggr|_{\gamma_{n_s} = \lambda}
\end{align}
where we have used $ \prod\nolimits_{i < j}^{n_s} {\left( {\gamma _j
- \gamma _i } \right)}  = \det \left( {\gamma _i^{j - 1} } \right)$.
To evaluate the $n_s - 1$ integrals, we expand $\det \left(
{{\bf{\Delta }}_1 } \right)$ along its last column and $\det \left(
{\gamma _i^{j - 1} } \right)$ along its last row, and then integrate
term-by-term by virtue of \cite[Lemma 2]{Shin06}. This yields
%
\begin{align}\label{eq:unorderedPDF3}
f_{\left. \lambda  \right|{\bf{L }}} \left( \lambda  \right) =
\frac{{\sum\limits_{l = 1}^{q} {\sum\limits_{k = q- n_s + 1}^{q}
{\beta _l^{q - n_s - 1} e^{ - \lambda /\beta _l } \lambda ^{q- n_s + k - 1}
\bar{D}_ {l, k} } } }}{{n_s\prod\nolimits_{i = 1}^{n_s} { \Gamma
\left( {n_s  - i + 1} \right) } \prod\nolimits_{i < j}^{q} {\left(
{\beta _j - \beta _i } \right)} }}
\end{align}
where ${\bar{D}_ {l,k} }$ is the $(l,k)$th cofactor of a $q \times
q$ matrix $ {\bf{\Xi }} = \bigl[ \mathbf{A} \; \mathbf{C}
\bigr]$, with entries
\begin{align}
{\begin{array}{*{20}c}
   {\left\{ {\bf{A}} \right\}_{m,n}  = \beta _m^{n - 1} } & {m = 1, \ldots ,q,} & {n = 1, \ldots ,q - n_s}  \\
\end{array}}
\end{align}
and
\begin{align}
\begin{array}{*{20}c}
   {\left\{ {\bf{C}} \right\}_{m,n}  = \Gamma \left( n \right)\beta _m^{q - n_s + n - 1} } & {m = 1, \ldots ,q,} & {n = 1, \ldots ,n_s}  \\
\end{array} \; .
\end{align}
Then, it can be shown that
\begin{align}
\sum\limits_{l = 1}^{q} {\sum\limits_{k = q - n_s + 1}^{q} {\beta
_l^{q - n_s - 1} e^{ - y/\beta _l } \lambda ^{q - n_s + k - 1}
\bar{D}_ {l,k} } } = {\sum\limits_{k = q - n_s + 1}^{q} {\det \left(
{{\bf{D}}_k } \right)} },
\end{align}
where ${{\bf{D}}_k }$ is a $q \times q$ matrix with entries
\begin{align}
\left\{ {{\bf{D}}_k } \right\}_{m,n}  = \left\{
{\begin{array}{*{20}c}
   {\beta _m^{n - 1} ,} & {m = 1, \ldots ,q,} & {n = 1, \ldots ,q - n_s} & {}  \\
   {\Gamma \left( n - q + n_s - 1  \right)\beta _m^{n } ,} & {m = 1, \ldots ,q,} & {n = q - n_s + 1, \ldots ,q,} & {n \ne k}  \\
   {\beta _m^{q - n_s - 1} e^{ -  \lambda/\beta _m } \lambda ^{n - q + n_s - 1} }, & {m = 1, \ldots ,q,} & {n = k} & {}  \\
\end{array}} \right. \;
\end{align}
Hence, we can rewrite (\ref{eq:unorderedPDF3}) as follows
\begin{align}\label{eq:unorderedPDF6}
f_{\left. \lambda  \right|{\bf{L }}} \left( \lambda  \right) =
\frac{{\sum\limits_{k = q - n_s + 1}^{q} {\det \left( {{\bf{D}}_k }
\right)} }}{{n_s\prod\nolimits_{i = 1}^{n_s} { \Gamma \left( {n_s -
i + 1} \right) } \prod\nolimits_{i < j}^{q} {\left( {\beta _j -
\beta _i } \right)} }} \; .
\end{align}
After some basic manipulations, (\ref{eq:unorderedPDF6}) can be
further simplified as
\begin{align}\label{eq:unorderedPDF7}
f_{\left. \lambda  \right|{\bf{L }}} \left( \lambda  \right) =
\frac{1}{{n_s \prod\nolimits_{i < j}^{q } {\left( {\beta _j  - \beta
_i } \right)} }}\sum\limits_{k = q  - n_s  + 1}^{q } {\frac{{\lambda
^{n_s  - q  + k - 1} }}{{ \Gamma \left( {n_s  - q + k} \right) }}}
\det \left( {{\bf{\bar D}}_k } \right)
\end{align}
where ${{\bf{\bar D}}_k }$ is a $q \times q$ matrix with entries
\begin{align}
\left\{ {{\bf{\bar D}}_k } \right\}_{m,n}  = \left\{
{\begin{array}{*{20}c}
   {\beta _m^{n - 1} }, & {n \ne k,}  \\
   {e^{ - \lambda /\beta _m } \beta _m^{q - n_s + 1}}, & {n = k.}  \\
\end{array}} \right. \;
\end{align}
Finally, we apply Laplace's expansion to (\ref{eq:unorderedPDF7}) to
yield the desired result.

\subsection{Proof of Lemma \ref{jointpdf_beta}}\label{sec:Proof_Jointpdf_beta}
The joint p.d.f. of ${\bf{W}}_1 = {\rm diag}\left\{ {\alpha _1 ,
\ldots ,\alpha _q } \right\}$ is given by
\cite{James64,Khatri66,Ratnarajah03}
\begin{align}\label{eq:fjointpdf}
f_{\mathbf{W}_1} \left( \alpha_1, \cdots, \alpha_q \right) &=
\mathcal{K}e^{ - \sum\limits_{i = 1}^{q } {\alpha _i }
}\prod\limits_{i = 1}^{q } {\alpha _i^{p  - q} } \prod\limits_{i <
j}^{q } {\left( {\alpha _j - \alpha _i } \right)^2 }.
\end{align}
Recalling that
\begin{align}\label{eq:alphai}
\alpha _i  = \frac{{\beta _i }}{{1 - a\beta _i }}
\end{align}
we derive the joint p.d.f. of $ {\bf{W}}_2  = {\rm diag}\left\{
{\beta _1 , \ldots ,\beta _q } \right\}$ from (\ref{eq:fjointpdf})
by applying a vector transformation \cite{Muirhead82}
\begin{align}\label{eq:gjointpdf}
f_{{\bf{W}}_2 } \left( \beta_1, \cdots, \beta_q  \right) =
f_{\mathbf{W}_1} \left( \frac{{\beta _1 }}{{1 - a\beta _1 }},
\cdots, \frac{{\beta _q }}{{1 - a\beta _q }} \right)\left|
{{\bf{J}}\left( {\left( {\alpha _1 , \ldots ,\alpha _q } \right) \to
\left( {\beta _1 , \ldots ,\beta _q } \right)} \right)} \right|,
\end{align}
where
\begin{align}\label{eq:expjacobian}
{\bf{J}}\left( {\left( {\alpha _1 , \ldots ,\alpha _q } \right) \to
\left( {\beta _1 , \ldots ,\beta _q } \right)} \right) = \det \left[
{\begin{array}{*{20}c}
   {\frac{{\partial \alpha _1 }}{{\partial \beta _1 }}} &  \cdots  & {\frac{{\partial \alpha _1 }}{{\partial \beta _q }}}  \\
    \vdots  &  \ddots  &  \vdots   \\
   {\frac{{\partial \alpha _q }}{{\partial \beta _1 }}} &  \cdots  & {\frac{{\partial \alpha _q }}{{\partial \beta _q }}}  \\
\end{array}} \right] \; .
\end{align}
From (\ref{eq:alphai}), we have
\begin{align}\label{eq:alphaidiff}
\frac{{\partial \alpha _i }}{{\partial \beta _i }} =
\frac{1}{{\left( {1 - a\beta _i } \right)^2 }},
\end{align}
therefore the Jacobian transformation in (\ref{eq:expjacobian})
is evaluated as
\begin{align}\label{eq:expjacobian1}
{\bf{J}}\left( {\left( {\alpha _1 , \ldots ,\alpha _q } \right) \to
\left( {\beta _1 , \ldots ,\beta _q } \right)} \right) =
\prod\limits_{i = 1}^q {\frac{1}{{\left( {1 - a\beta _i } \right)^2
}}}.
\end{align}
Substituting (\ref{eq:fjointpdf}) and (\ref{eq:expjacobian1}) into
(\ref{eq:gjointpdf}) yields
\begin{align} \label{eq:jointPDF1}
f_{{\bf{W}}_2 } \left( \beta_1, \cdots, \beta_q  \right) =
\mathcal{K} \prod_{i=1}^{q} \frac{\beta_i^{p - q} e^{- \frac{
\beta_i}{1-a\beta_i}}}{(1-a\beta_i)^{p - q + 2}} \prod_{i<j}^q
\left( \frac{ \beta_j }{1-a\beta_j} - \frac{ \beta_i }{1-a\beta_i}
\right)^2.
\end{align}
Finally, simplifying using
\begin{align}
\prod_{i<j}^q \left( \frac{ \beta_j }{1-a\beta_j} - \frac{ \beta_i }{1-a\beta_i} \right)^2 &= \prod_{i<j}^q \left( \frac{ \beta_j - \beta_i }{( 1-a\beta_j ) ( 1-a\beta_i ) } \right)^2 \nonumber \\
&= \frac{ \prod_{i<j}^q (\beta_j - \beta_i )^2 } { \prod_{i=1}^q (
1- a \beta_i)^{2(q-1)} }  \;
\end{align}
yields the joint p.d.f. of $\bf{L}$.

We now derive the p.d.f.\ of an unordered eigenvalue $\beta$ of the
diagonal matrix $\bf{L}$. According to \cite[Eq. 42]{Shin03}, the
unordered eigenvalue p.d.f. of ${\bf{H}}_2 {\bf{H}}_2^\dag$ is given
by
\begin{align}\label{eq:expjacobian2}
f\left( \lambda  \right) = \frac{1}{q} \sum\limits_{i = 0}^{q - 1}
\sum\limits_{j = 0}^i \sum\limits_{l = 0}^{2j} \mathcal{A}\left(
{i,j,l,p,q } \right) \lambda ^{p - q + l} e^{ - \lambda } \; .
\end{align}
Recalling that $\beta = \lambda /\left( {1 + a\lambda } \right)$,
the result follows after applying a simple transformation.

\subsection{Proof of Theorem \ref{final_unorderedpdf}}\label{sec:Proof_finalunorderedpdf}

We start by re-expressing the conditional unordered eigenvalue
p.d.f.\ $f_{\lambda | \mathbf{L}} ( \cdot )$ in \textit{Lemma
\ref{unorderedpdf}} as follows
\begin{align}\label{eq:conunorder1}
f_{\left. \lambda  \right|{\bf{L}}} \left( \lambda  \right) =
\frac{1}{{s\prod\nolimits_{i < j}^q {\left( {\beta _j  - \beta _i }
\right)} }}\sum\limits_{k = q - s + 1}^q {\frac{{\lambda ^{n_s  - q
+ j - 1} }}{{\Gamma \left( {n_s  - q + j} \right)}}} \det \left(
{{\bf{\tilde D}}_k } \right),
\end{align}
where ${{\bf{\tilde D}}_k }$ is a $q \times q$ matrix with entries
\begin{align}
\left\{ {{\bf{\tilde D}}_k } \right\}_{m,n}  = \left\{
{\begin{array}{*{20}c}
   {\beta _m^{n - 1} }, & {n \ne k,}  \\
   {e^{ - \lambda /\beta _m } \beta _m^{q - n_s  - 1} }, & {n = k.}  \\
\end{array}} \right.
\;
\end{align}
Now, utilizing \textit{Lemma \ref{jointpdf_beta}}, we can evaluate
the unconditional p.d.f.\ as follows
\begin{align}\label{eq:flamdaexp}
f_{\lambda } ( \lambda ) &= E_{\mathbf{L}} \left[ f_{\lambda | \mathbf{L}} ( \lambda ) \right] \nonumber \\
&=  \frac{ \mathcal{K} }{s } \sum\limits_{k = q - s + 1}^q
{\frac{{\lambda ^{n_s - q + k - 1} }}{{\Gamma \left( {n_s  - q + k}
\right)}}\mathcal{\bar I}_k }
\end{align}
where
\begin{align}\label{eq:barIk}
\mathcal{\bar I}_k &= \int_{0 \leq \beta_1 < \cdots < \beta_q \leq
1/a} \det( {{\bf{\tilde D}}_k } ) { \prod\limits_{i < j}^q {\left(
{\beta _j  - \beta _i } \right)}}\prod_{l=1}^q \frac{\beta_l^{p - q} e^{-\frac{\beta_l}{1-a\beta_l}}}{(1-a\beta_l)^{p + q} } d \beta_1 \cdots d \beta_q \nonumber \\
&= \det( \mathbf{\tilde{Y}}_k ),
\end{align}
where $\mathbf{\tilde{Y}}_k$ is a $q \times q$ matrix with entries
\begin{align}\label{eq:ykij}
\{ \mathbf{\tilde{Y}}_k \}_{m,n} = \left\{
\begin{array}{lr}
\int_{0}^{1/a} \frac{x^{p - q + m + n - 2}}{(1-ax)^{p+q}} e^{-\frac{x}{1-ax}} d x,  &  \; \, n \neq k, \\
\int_{0}^{1/a} \frac{x^{p - n_s + m - 2}}{(1-ax)^{p+q}}
e^{-\frac{x}{1-ax}} e^{-\lambda/x} d x, &  \; \, n = k.
\end{array}
\right. \; \;
\end{align}
Let $ t = x/\left( {1 - ax} \right)$. Utilizing \cite[Eq.
3.383.5]{Gradshteyn00} and \cite[Eq. 3.471.9]{Gradshteyn00}, the
integrals in (\ref{eq:ykij}) can be evaluated, respectively, as\footnote{Note that, by using the Binomial expansion, (\ref{eq:ykij1}) can be alternatively expressed as
$$
\int_0^\infty
{t^{p - q + m + n - 2} \left( {1 + at} \right)^{ 2q - m
- n} e^{ - t} } dt = \sum\limits_{i = 0}^{2q - m - n} {a^i \Gamma \left( {p - q + m + n + i - 1} \right)} \; .
$$
}
\begin{align}\label{eq:ykij1}
\int_0^{1/a} {\frac{{x^{p - q + m + n - 2} }}{{\left( {1 - ax}
\right)^{p + q} }}e^{ - \frac{x}{{1 - ax}}} } dx &= \int_0^\infty
{t^{p - q + m + n - 2} \left( {1 + at} \right)^{ 2q - m
- n} e^{ - t} } dt\nonumber\\
&= {a^{q - p - m - n + 1} \Gamma \left( {p - q + m + n - 1}
\right)U\left( {p - q + m + n - 1,p + q,1/a} \right)}
\end{align}
and
\begin{align}\label{eq:ykij2}
& \int_0^{1/a} {\frac{{x^{p - n_s + m - 2} }}{{\left( {1 - ax}
\right)^{p + q} }}e^{ - \frac{x}{{1 - ax}}} } e^{ - \lambda /x} dx
\nonumber \\
& \hspace*{2cm} = e^{ - \lambda a} \int_0^\infty  {t^{p - n_s + m -
2} \left( {1 + at}
\right)^{ q + n_s - m } e^{ - t - \lambda /t} } dt\nonumber\\
& \hspace*{2cm} = e^{ - \lambda a}\sum\limits_{i = 0}^{q + n_s - m }
{ \binom{q+n_s-m}{i} a^{q + n_s - m - i} \int_0^\infty  {t^{p + q -
i - 2} e^{ - t - \lambda /t} } dt}
\nonumber\\
& \hspace*{2cm} = 2e^{ - \lambda a}\sum\limits_{i = 0}^{q + n_s - m
} { \binom{q+n_s-m}{i}
%
a^{q + n_s - m - i } } \lambda ^{\left( {p + q - i - 1} \right)/2}
K_{p + q - i - 1} \left( {2\sqrt \lambda  } \right),
\end{align}
where $ U\left( { \cdot , \cdot , \cdot } \right)$ is the confluent
hypergeometric function of the second kind \cite[Eq.
9.211.4]{Gradshteyn00}.

Combining (\ref{eq:flamdaexp})--(\ref{eq:ykij2}) and then applying
Laplace's expansion yields the desired result.

\subsection{Proof of Lemma \ref{expdet}}\label{sec:Proof_expdet}

We will prove the lemma by giving a separate treatment for the two
cases, $q < n_s$ and $ q \ge n_s$.

\subsubsection{$q < n_s$ Case} In this case, we start by writing
\begin{align} \label{eq:FirstEq}
E\left\{ {\left. {\det \left( {{\bf{I}}_{n_s }  + \frac{{\rho
a}}{{n_s }}{\bf{\tilde H}}_1^\dag  {\bf{L\tilde H}}_1 } \right)}
\right|{\bf{L}}} \right\} &= E\left\{ {\left. {\det \left(
{{\bf{I}}_{q }  + \frac{{\rho a}}{{n_s }} {\bf{L\tilde H}}_1
{\bf{\tilde H}}_1^\dag } \right)} \right|{\bf{L}}} \right\}
\nonumber \\
&= E \left\{ \prod_{i=1}^q \left( 1 + \frac{{\rho a}}{{n_s }}
\gamma_i  \right) \bigg|{\bf{L}} \right\}
\end{align}
where $\gamma _1 , \ldots ,\gamma _q $ are the ordered eigenvalues
of ${\bf{L\tilde H}}_1 {\bf{\tilde H}}_1^\dag$. Conditioned on
$\mathbf{L}$, the joint p.d.f.\ of $\gamma _1 , \ldots ,\gamma _q $
is given in \cite{Chiani03}. Using this result, we can express
(\ref{eq:FirstEq}) as follows
\begin{align}\label{eq:upperbound1}
E\left\{ {\left. {\det \left( {{\bf{I}}_{n_s }  + \frac{{\rho
a}}{{n_s }}{\bf{\tilde H}}_1^\dag  {\bf{L\tilde H}}_1 } \right)}
\right|{\bf{L}}} \right\} & \nonumber \\
& \hspace*{-4cm} = \frac{{\int_{\mathcal{D}_{\rm ord}} {\det \left(
{e^{ - \gamma _j /\beta _i } } \right)\prod\nolimits_{i = 1}^q
{\left( {1 + \frac{{\rho a }}{{n_s }}\gamma _i } \right)\beta _i^{q
- n_s  - 1} \gamma _i^{n_s  - q} \det( \gamma_i^{j-1} ) d\gamma _1
\cdots d\gamma _q } } }}{\prod\nolimits_{i = 1}^q {\Gamma \left(
{n_s  - i + 1} \right)} \prod\nolimits_{i < j}^q {\left( {\beta _j -
\beta _i } \right)} }
\end{align}
where the integrals are taken over the region $\mathcal{D}_{\rm ord}
= \{\infty \geq \gamma_1 \geq \cdots \gamma_q \geq 0\}$.
Applying \cite[Corollary 2]{Chiani03}, (\ref{eq:upperbound1}) can be
evaluated in closed-form as follows
\begin{align} \label{eq:SecEq}
E\left\{ {\left. {\det \left( {{\bf{I}}_{n_s }  + \frac{{\rho
a}}{{n_s }}{\bf{\tilde H}}_1^\dag  {\bf{L\tilde H}}_1 } \right)}
\right|{\bf{L}}} \right\} = \frac{{\prod\nolimits_{i = 1}^q {\beta
_i^{q - n_s  - 1} } \det \left( {{\bf{\Xi }}_1 } \right)}}
{\prod\nolimits_{i = 1}^q {\Gamma \left( {n_s  - i + 1} \right)}
\prod\nolimits_{i < j}^q {\left( {\beta _j  - \beta _i } \right)} },
\end{align}
where $ {\bf{\Xi }}_1 $ is a $q \times q$ matrix with entries
\begin{align}
\left\{ {{\bf{\Xi }}_1 } \right\}_{m,n} = \beta _m^{n_s - q + n}
\left( {\Gamma \left( {n_s  - q + n} \right) + \frac{{\rho a}}{{n_s
}}\beta _m \Gamma \left( {n_s  - q + n + 1} \right)} \right).
\end{align}
Extracting common factors from the determinant in (\ref{eq:SecEq})
and simplifying yields the desired result.

\subsubsection{$q \ge n_s$ Case} In this case, we use the joint eigenvalue p.d.f.\ (\ref{eq:unorderedPDF}) to obtain
\begin{align}\label{eq:upperbound2}
E\left\{ {\left. {\det \left( {{\mathbf{I}}_{n_s} + \frac{{\rho
a}}{{n_s }}{\mathbf{\tilde H}}_1^\dag  {\mathbf{L\tilde H}}_1 }
\right)} \right|{\mathbf{L}}} \right\} &= E \left\{
\prod_{i=1}^{n_s} \left( 1 +
\frac{{\rho a}}{{n_s }} \gamma_i \right) \bigg|{\bf{L}}  \right\} \nonumber \\
&= \frac{{\int_{ \mathcal{D}_{\rm ord}} {\prod\nolimits_{i = 1}^{n_s
} {\left( {1 + \frac{{\rho a }} {{n_s }}\gamma _i } \right)\det
\left( {{\mathbf{\Delta }}_1 } \right) \det( \gamma_i^{j-1} )
d\gamma _1 \cdots d\gamma _{n_s } } } }} {\prod\nolimits_{i =
1}^{n_s} {\Gamma \left( {n_s  - i + 1} \right)} \prod\nolimits_{i <
j}^q {\left( {\beta _j - \beta _i } \right)} },
\end{align}
where $\gamma _1 , \ldots ,\gamma _{n_s }$ are the ordered
eigenvalues of ${\mathbf{\tilde H}}_1^\dag  {\mathbf{L\tilde H}}_1$,
$\mathbf{\Delta}_1$ is defined in (\ref{eq:Delta1Defn}), and the
integration region is $\mathcal{D}_{\rm ord} = \{\infty \geq
\gamma_1 \geq \cdots \gamma_{n_s} \geq 0\}$.
Applying \cite[Lemma 2]{Shin06}, (\ref{eq:upperbound2}) can
evaluated in closed-form as follows
\begin{align}\label{eq:detIplus}
E\left\{ {\left. {\det \left( {{\mathbf{I}}_{n_s} + \frac{{\rho
a}}{{n_s }}{\mathbf{\tilde H}}_1^\dag  {\mathbf{L\tilde H}}_1 }
\right)} \right|{\mathbf{L}}} \right\}= \frac{{\det \left(
{{\mathbf{\Xi }}_2 } \right)}} {\prod\nolimits_{i = 1}^{n_s} {\Gamma
\left( {n_s  - i + 1} \right)} \prod\nolimits_{i < j}^q {\left(
{\beta _j  - \beta _i } \right)} },
\end{align}
where $ {\bf{\Xi }}_2  = \bigl[ {{\bf{A}}_1 } \; \;  {{\bf{C}}_1 }
\bigr]$
is a $q \times q$ matrix with entries
\begin{align}
{\begin{array}{*{20}c}
   {\left\{ {{\bf{A}}_1} \right\}_{m,n}  = \beta _m^{n - 1} },  & {n = 1, \ldots ,q - n_s}  \\
\end{array}}
\end{align}
and
\begin{align}
\left\{ {{\bf{C}}_1 } \right\}_{m,n}  &= \beta _m^{n+q-n_s-1} \left(
{\Gamma \left( n \right) + \left( {\rho a/n_s } \right)\beta _m } {
\Gamma \left( n + 1 \right)} \right), \hspace*{0.6cm} n = 1, \ldots,
n_s.
\end{align}
Extracting common factors from ${\det \left( {{\mathbf{\Xi }}_2 }
\right)}$ and simplifying yields the desired result.

\subsection{Proof of Lemma \ref{lnexpdet}}\label{sec:Proof_lnexpdet}

To prove this lemma, it is convenient give a separate treatment for
the two cases, $q < n_s$ and $ q \ge n_s$.

\subsubsection{$q < n_s$ Case}

Now we need to calculate the expectation $ E\left\{ {\ln \det
\left( {{\bf{ L\tilde H}}_1 {\bf{\tilde H}}_1^\dag } \right)}
\right\}$. The moment generating function (m.g.f.) of $\ln \det \left( {{\bf{ L\tilde H}}_1
{\bf{\tilde H}}_1^\dag  } \right)$, conditioned on $\mathbf
{ L}$, is given by
\begin{align}\label{eq:mgf0}
\mathcal{M}_1 \left( {t\left| {\bf{ L}} \right.} \right) = E\left\{
{\left. {\det \left( {{\bf{ L\tilde H}}_1 {\bf{\tilde H}}_1^\dag  }
\right)^t } \right|{\bf{ L}}} \right\} \; .
\end{align}
Utilizing the joint p.d.f. of the eigenvalues $ \gamma _1 , \ldots , \gamma _q $ of $ {{\bf{ L\tilde
H}}_1 {\bf{\tilde H}}_1^\dag  }$, presented in
\cite{Alfano04,Chiani03}, we get
\begin{align}\label{eq:mgf01}
\mathcal{M}_1 \left( {t\left| {\bf{ L}} \right.} \right) =
\frac{{\int_{ \mathcal{F}_{\rm ord}} {\det \left( {e^{ -
\gamma _j / \beta _i } } \right)\prod\nolimits_{i = 1}^q
{ \gamma _i^{n_s  - q + t}  \beta _i^{q - n_s  - 1} }
\prod\nolimits_{i < j}^q {\left( {\gamma _j  - \gamma
_i } \right)} d \gamma _1  \cdots d \gamma _q }
}}{{\prod\nolimits_{i = 1}^q {\Gamma \left( {n_s  - i + 1} \right)}
\prod\nolimits_{i < j}^q {\left( { \beta _j  - \beta _i
} \right)} }}
\end{align}
where the integrals are taken over the region $\mathcal{F}_{\rm ord}
= \{\infty \geq \gamma_1 \geq \cdots  \gamma_q \geq 0\}$. Applying \cite[Corollary 2]{Chiani03},
(\ref{eq:mgf01}) can be further simplified as follows
\begin{align}\label{eq:mgf02}
\mathcal{M}_1 \left( {t\left| {\bf{ L}} \right.} \right) =
\frac{{\det \left( {{\bf{\Xi }}_3 } \right)}}{{\prod\nolimits_{i =
1}^q {\Gamma \left( {n_s  - i + 1} \right)} \prod\nolimits_{i < j}^q
{\left( {\beta _j  - \beta _i } \right)} }}
\end{align}
where ${{\bf{\Xi }}_3 }$ is a $q \times q$ matrix with entries
\begin{align}
\left\{ {{\bf{\Xi }}_3 } \right\}_{m,n}  =  \beta _m^{q - n_s  - 1}
\int_0^\infty  {e^{ - y/ \beta _m } y^{n_s  - q + t + n - 1} dy}  =
\beta _i^{t + n - 1} \Gamma \left( {n_s  - q + t + n} \right) \; .
\end{align}
From $\mathcal{M}_1 \left( {t\left| {\bf{ L}} \right.} \right)$, we
get
\begin{align}\label{eq:lndetmgf1}
E\left\{ {\left. {\ln \det \left( {{\bf{ L\tilde H}}_1 {\bf{\tilde H}}_1^\dag  } \right) } \right|{\bf{ L}}} \right\} &= \left. {\frac{d}{{dt}}\mathcal{M}_1 \left( {t\left| {\bf{ L}} \right.} \right)} \right|_{t = 0} \nonumber\\
& = \frac{{\sum\limits_{k = 1}^q {\det \left( {{\bf{\Sigma }}_k }
\right)} }}{{\prod\nolimits_{i = 1}^q {\Gamma \left( {n_s  - i + 1}
\right)} \prod\nolimits_{i < j}^q {\left( { \beta _j  -  \beta _i }
\right)} }}
\end{align}
where ${{\mathbf{\Sigma }}_k }$ is a $q \times q$ matrix whose
entries are
\begin{align}
\left\{ {{\bf{\Sigma }}_k } \right\}_{m,n}  = \left\{
{\begin{array}{*{20}c}
   { \beta _m^{n - 1} \Gamma \left( {n_s  - q + n} \right)}, & {n \ne k,}  \\
   { \beta _m^{n - 1} \Gamma \left( {n_s  - q + n} \right)\left[ {\psi \left( {n_s  - q + n} \right) + \ln  \beta _m } \right]}, & {n = k.}  \\
\end{array}} \right.
\end{align}
where $\psi(\cdot)$ is the digamma function.
Now, ${\det \left( {{\mathbf{\Sigma }}_k } \right)}$ can be further
simplified as
\begin{align}\label{eq:detomegaj}
\det \left( {{\bf{\Sigma }}_k } \right) = \det \left( {{\bf{\tilde
\Sigma }}_k } \right) \prod\limits_{k = 1}^{q } {\Gamma \left( n_s - q + k
\right)}
\end{align}
where ${{\bf{\tilde \Sigma }}_k }$ is a $q \times q$ matrix with
entries
\begin{align}
\left\{ {{\bf{\tilde \Sigma }}_k } \right\}_{m,n}  = \left\{
{\begin{array}{*{20}c}
   {\beta _m^{n - 1} }, & {n \ne k,}  \\
   {\beta _m^{n - 1} \left[ {\psi \left( {n_s  - q + n} \right) + \ln \beta _m } \right]}, & {n = k.}  \\
\end{array}} \right. \; \; \;
\end{align}
By using the multi-linear property of determinants, along with some
basic manipulations, we can write
\begin{align}\label{eq:detomegaj1}
\det \left( {{\bf{\tilde \Sigma }}_k } \right) = \psi \left( {n_s  -
q + k} \right) \det \left( \beta_i^{j-1} \right) + \det \left(
{{\bf{Y}}_k } \right) \; .
\end{align}
Substituting (\ref{eq:detomegaj}) and (\ref{eq:detomegaj1}) into
(\ref{eq:lndetmgf1}) and simplifying yields the desired result.


\subsubsection{$q \ge n_s $ Case}

We now evaluate the m.g.f. of ${\ln\det \left( {{\mathbf{\tilde
H}}_1^\dag {\mathbf{L\tilde H}}_1 } \right)}$, conditioned on
$\mathbf {L}$, which is given by
\begin{align}\label{eq:mgf1}
\mathcal{M}_2\left( {\left. t \right|{\mathbf{L }}} \right) =
E\left\{ {\left. {\det \left( {{\bf{\tilde H}}_1^\dag  {\bf{L\tilde
H}}_1 } \right)^t } \right|{\bf{L}}} \right\}.
\end{align}
Utilizing (\ref{eq:unorderedPDF}), (\ref{eq:mgf1}) can be expressed
as
\begin{align}\label{eq:mgf2}
\mathcal{M}_2\left( {\left. t \right|{\mathbf{L }}} \right) =
\frac{1}{{\prod\nolimits_{i = 1}^{n_s} {\Gamma \left( {n_s  - i + 1} \right)}
\prod\nolimits_{i < j}^{q} {\left( {\beta _j  - \beta _i } \right)}
}}\int_{\mathcal{D}_{\rm ord}} {\prod\limits_{i = 1}^{n_s} {\gamma
_i^t } \det \left( {{\bf{\Delta }}_2 } \right) \det( \gamma_i^{j-1}
) } d\gamma _1 , \ldots ,d\gamma _{n_s},
\end{align}
where $\mathcal{D}_{\rm ord} = \{\infty \geq \gamma_1 \geq \cdots
\gamma_{n_s} \geq 0\}$. Applying \cite[Lemma 2]{Shin06} yields
\begin{align}\label{eq:mgf3}
\mathcal{M}_2\left( {\left. t \right|{\mathbf{L }}} \right) =
\frac{{\det \left( {{\bf{\Xi }}_4 } \right)}}{{\prod\nolimits_{i =
1}^{n_s } {\Gamma \left( {n_s  - i + 1} \right)
} \prod\nolimits_{i < j}^{q }
{\left( {\beta _j  - \beta _i } \right)} }},
\end{align}
where $ {\bf{\Xi }}_4  = \bigl[ {{\bf{A}}_2 } \; \;  {{\bf{C}}_2 }
\bigr]$
%
is a $q \times q$ matrix with entries
\begin{align}
{\begin{array}{*{20}c}
   {\left\{ {\bf{A}}_2 \right\}_{m,n}  = \beta _m^{n - 1} }, \; \;  & {n = 1, \ldots ,q - n_s}  \\
\end{array}}
\end{align}
and
\begin{align}
\begin{array}{*{20}c}
   {\left\{ {\bf{C}}_2 \right\}_{m,n}  = \Gamma \left( t + n \right)\beta _m^{q - n_s + t + n - 1} }, \; \;  & {n = 1, \ldots ,n_s}  \\
\end{array}
\end{align}
From the m.g.f.\ (\ref{eq:mgf3}), we can then obtain
\begin{align}\label{eq:mgfcondL}
E\left\{ {\left. {\ln \det \left( {{\mathbf{\tilde H}}_1^\dag
{\mathbf{L\tilde H}}_1 } \right)} \right|{\mathbf{L}}} \right\} &=
\left. {\frac{d} {{dt}}\mathcal{M}_2\left( {\left. t
\right|{\mathbf{L}}} \right)} \right|_{t = 0}\nonumber\\
& = \frac{{\sum\limits_{k = q - n_s + 1 }^q {\det \left(
{{\mathbf{\Omega }}_k } \right)} }} {{\prod\nolimits_{i = 1}^{n_s }
{\Gamma \left( {n_s  - i + 1} \right)} \prod\nolimits_{i < j}^{q}
{\left( {\beta _j  - \beta _i } \right)} }}
\end{align}
where ${{\mathbf{\Omega }}_k }$ is a $q \times q$ matrix with
entries
\begin{align}
\left\{ {{\bf{\Omega }}_k } \right\}_{m,n}  = \left\{
{\begin{array}{*{20}c}
   {\beta _m^{n - 1} }, & {n \ne k,} & {n = 1, \ldots ,q - n_s, }  \\
   {\Gamma \left( n_s - q + n  \right)\beta _m^{n - 1} }, & {n \ne k,} & {n = q - n_s  + 1, \ldots ,q_s, }  \\
   {\beta _m^{n - 1} \Gamma \left( n_s - q + n \right)\left[ {\psi \left( n_s - q + n \right) + \ln \beta _m } \right]}, & {n = k.}  \\
\end{array}} \right. \;
\end{align}
By using the multi-linear property of determinants, along with some
basic manipulations, we can obtain the desired result.

\subsubsection{$q = s$ Case}
In this case, starting with (\ref{eq:lnexpdet}), we can write the
determinant summation over $k$ as follows
\begin{align}
\sum\limits_{k = 1}^q {\det \left( {{\bf{Y}}_k } \right)}  &=
\sum\limits_{k = 1}^q {\sum\limits_{ \left\{ \alpha \right\} } {{\rm
sgn}(\alpha) \left[ {\prod\limits_{i = 1}^q {\beta _{\alpha \left( i
\right)}^{i - 1} } } \right]} } \ln \beta _{\alpha \left( k \right)}
\end{align}
where the second summation is over all permutations $\alpha = \{
\alpha \left( 1 \right), \ldots ,\alpha \left( q \right) \}$ of the
set $\{ 1, \ldots, q \}$, with ${\rm sgn}(\alpha)$ denoting the sign
of the permutation. We can further write
\begin{align} \label{eq:lnexpdet_sp2}
\sum\limits_{k = 1}^q {\det \left( {{\bf{Y}}_k } \right)}  &=
\sum\limits_{ \left\{ \alpha \right\}  } { {\rm sgn}(\alpha)
  \left[
{\prod\limits_{i = 1}^q {\beta _{\alpha \left( i \right)}^{i - 1} }
} \right]} \sum\limits_{k = 1}^q {\ln
\beta _{\alpha \left( k \right)} }\nonumber\\
& = \ln \det \left( { {\rm diag} \left\{ {\beta _i } \right\}_{i =
1}^q } \right){{\prod\nolimits_{i < j}^q {\left( {\beta _j - \beta
_i } \right)} }} \nonumber \\
& = \ln \det \left( \mathbf{L} \right){{\prod\nolimits_{i < j}^q
{\left( {\beta _j - \beta _i } \right)} }} \; .
\end{align}
Substituting (\ref{eq:lnexpdet_sp2}) into (\ref{eq:lnexpdet}) yields
the final result.

\subsection{Proof of Theorem \ref{lnexpdetNew}}\label{sec:Proof_lnexpdetNew}


We start with \textit{Lemma \ref{lnexpdet}} and remove the
conditioning on $\mathbf{L}$ by using \textit{Lemma
\ref{jointpdf_beta}} as follows
\begin{align}\label{eq:expunconL}
& E\left\{ {\ln \det \left( {\bf{\Phi }} \right)} \right\} =
%
\sum\limits_{k = 1}^s {\psi \left( {n_s -s + k } \right)}
\nonumber\\
&{~~~~}+ {\mathcal{K}} \int_{0 < \beta _1 < \cdots < \beta _q  \le
1/a} { \det \left(\beta_i^{j-1} \right)  \prod\limits_{i = 1}^q
{g\left( {\beta _i } \right)} \sum\limits_{k = q - n_s + 1}^q {\det
\left( {{\mathbf{Y }}_k } \right)} d\beta _1 \cdots d\beta _q },
\end{align}
where
\begin{align}
g\left( u \right) = \frac{{u^{p - q} e^{ - u/\left( {1 - au}
\right)} }} {{\left( {1 - au} \right)^{p + q} }} \; \; .
\end{align}
Using \cite[Lemma 2]{Shin06}, these integrals can be simplified to
give
\begin{align} \label{eq:UncondLnDet}
E\left\{ {\ln \det \left( {\bf{\Phi }} \right)} \right\}
=
\sum\limits_{k = 1}^s {\psi \left( {n_s -s + k } \right)}  +
{\mathcal{K}} \sum\limits_{k = q - n_s + 1}^q {\det \left( {{
{{\bf{\tilde W}}}}_k } \right)},
\end{align}
where $ {{{{\bf{\tilde W}}}}_k }$ is a $q \times q$ matrix with
entries
\begin{align}\label{eq:tildeW}
\left\{ {{\bf{\tilde W}}_k } \right\}_{m,n}  = \left\{
{\begin{array}{*{20}c}
   {\int_0^{1/a} {\frac{{u^{p - q + m + n  - 2} }}{{\left( {1 - au} \right)^{p + q} }}} e^{ - \frac{u}{{1 - au}}} du}, & {n \ne k,}  \\
   {\int_0^{1/a} {\frac{{u^{p - q + m + n  - 2} }}{{\left( {1 - au} \right)^{p + q} }}}  e^{ - \frac{u}{{1 - au}}} \ln u du}, & {n = k.}  \\
\end{array}} \right. \; \;
\end{align}
For the case $n \neq k$, a closed-form expression is given in
(\ref{eq:ykij1}).  For the case $n = k$, we utilize \cite[Eq.
4.358.5]{Gradshteyn00} and \cite[Eq. 47]{Shin03}, to obtain
\begin{align}\label{eq:integral2}
& \int_0^{1/a} {\frac{{u^{p - q + m + n - 2} }} {{\left( {1 - au}
\right)^{p + q} }}e^{ - \frac{u} {{1 - au}}} \ln u} du \nonumber \\
& \hspace*{2cm} = \int_0^\infty  {t^{p - q + m + n - 2} \left( {1 +
at} \right)^{2q - m - n} e^{ - t} \left[ {\ln t - \ln \left( {1 +
at} \right)} \right]dt}\nonumber\\
& \hspace*{2cm} =  \sum\limits_{i = 0}^{2q - m - n} {a^{2q - m - n -
i} \binom{2q - m - n}{i} \int_0^\infty  {t^{p + q - i - 2} e^{ - t}
\left[ {\ln t - \ln \left( {1 + at} \right)} \right]dt} } \nonumber
\\
& \hspace*{2cm} = \sum\limits_{i = 0}^{2q - m - n} a^{2q - m - n -
i} \binom{2q - m - n}{i} \Gamma \left( {p + q - i - 1} \right)
\nonumber \\
& \hspace*{5cm} \times \left[ {\psi \left( {p + q - i - 1} \right) -
e^{1/a} \sum\limits_{l = 0}^{p + q - i - 2} {E_{l + 1} \left(
{\frac{1} {a}} \right)} } \right] \; \; .
\end{align}
Substituting (\ref{eq:ykij1}) and (\ref{eq:integral2}) into
(\ref{eq:tildeW}) and (\ref{eq:UncondLnDet}) yields (\ref{eq:lndetsp1}).

When $q = s$, we start with (\ref{eq:lndetsp1}) and remove the
conditioning on ${\bf {L}}$ as follows
\begin{align}
E\left\{ {\ln \det \left( {\bf{\Phi }} \right)} \right\} =
\sum\limits_{k = 1}^q {\psi \left( {n_s  - q + k} \right)} +
q\int_0^\infty  f\left( {\bar \beta} \right) \ln {\bar \beta} d{\bar
\beta} \;
\end{align}
where $f\left( {\bar \beta}  \right)$ denotes the unordered
eigenvalue p.d.f.\ of ${\bf {L}}$ (i.e.\ p.d.f.\ of a
randomly-selected ${\bar \beta} \in \{\beta _1, \cdots , \beta _q
\}$). Substituting this p.d.f.\ from (\ref{eq:Unordered_Beta}) and
integrating using (\ref{eq:integral2}), we obtain the desired
result.

\section{Ergodic Capacity Proofs}

\subsection{Proof of Eq.\ (\ref{eq:exact_sp1})}\label{sec:Proof_exactsp1}
When $n_r \to \infty$, the ergodic capacity expression
(\ref{eq:ergcapa}) can be expressed as follows
\begin{align}\label{eq:caparay}
\lim_{n_r \to \infty }  {C\left( \rho  \right)} =
\frac{1}{2}E\left\{ {\log _2 \det \left( {{\bf{I}}_{n_s }  +
\frac{{\rho \alpha }}{{n_s \left( {1 + \rho } \right)}}{\bf{\tilde
H}}_1^\dag  {\bf{\tilde L}}_1 {\bf{\tilde H}}_1 } \right)} \right\}
\end{align}
where $ {\bf{\tilde L}}_1  = {\rm diag}\left\{ {\lambda _i^2 /\left(
{n_r \left( {1 + a\lambda _i^2 } \right)} \right)} \right\}$. Noting
that $q = n_d$, by the Law of Large Numbers we have
\begin{align}
\lim_{n_r \to \infty } {\frac{{{\bf{H}}_2 {\bf{H}}_2^\dag }}{{n_r
}}}   = {\bf{I}}_{n_d}
\end{align}
which implies that
\begin{align}\label{eq:centrallimit}
  \lim_{n_r \to \infty }  { {\frac{{\lambda _i^2 }}{{n_r }}} = 1} \;, & \; \hspace*{0.5cm}  {i = 1, \ldots,
  n_d} \; .
\end{align}
Recalling (\ref{eq:aDefn}), application of (\ref{eq:centrallimit})
in (\ref{eq:caparay}) yields
\begin{align} \label{eq:largeNRCap}
\lim_{n_r \to \infty } {C\left( \rho  \right)} = \frac{1}{2}E\left\{
{\log _2 \det \left( {{\bf{I}}_{n_s }  + \frac{{\rho \alpha
}}{{n_s(1 + \rho + \alpha)}}{\bf{H}}^\dag {\bf{H}} } \right)}
\right\},
\end{align}
where ${\bf{H}}$ is an $n_d \times n_s$ i.i.d.\ Rayleigh fading MIMO
channel matrix. Applying the identity (\ref{eq:DetProp}) to
(\ref{eq:largeNRCap}) yields the desired result.

\subsection{Proof of Eq.\ (\ref{eq:exact_sp2})}\label{sec:Proof_exactsp2}
Using (\ref{eq:DetProp}), the ergodic capacity expression
(\ref{eq:ergcapa}) can be alternatively written as
\begin{align}\label{eq:exactsp2capa}
C\left( \rho  \right) = \frac{1}{2}E\left\{ {\log _2 \det \left(
{{\bf{I}}_q  + \frac{{\rho a}}{{n_s }}{\bf{\tilde H}}_1 {\bf{\tilde
H}}_1^\dag  {\bf{L}}} \right)} \right\} \; .
\end{align}
By the Law of Large Numbers we have
\begin{align}
\lim_{n_s \to \infty} \frac{{{\bf{\tilde H}}_1 {\bf{\tilde
H}}_1^\dag }}{{n_s }} \to {\bf{I}}_q \;
\end{align}
and hence (\ref{eq:exactsp2capa}) reduces to
\begin{align}\label{eq:exactsp2capa1}
\lim_{n_s \to \infty} {C\left( \rho  \right)} = \frac{1}{2}E\left\{
{\log _2 \det \left( {{\bf{I}}_q + \rho a{\bf{L}}} \right)} \right\}
\; .
\end{align}
Substituting (\ref{eq:Ldefinition}) into (\ref{eq:exactsp2capa1}),
after some simple manipulations we easily obtain
\begin{align} \label{eq:LargeNRCap2}
\lim_{n_s \to \infty} {C\left( \rho  \right)} = \frac{1}{2}E\left\{
{\log _2 \det \left( {{\bf{I}}_q  + \left( {\rho + 1}
\right)a{\bf{H}}_2^\dag  {\bf{H}}_2 } \right)} \right\} -
\frac{1}{2}E\left\{ {\log _2 \det \left( {{\bf{I}}_q  +
a{\bf{H}}_2^\dag  {\bf{H}}_2 } \right)} \right\} \; .
\end{align}
Substituting (\ref{eq:aDefn}) into (\ref{eq:LargeNRCap2}) and
applying the identity (\ref{eq:DetProp}) yields the desired result.


\subsection{Proof of Theorem \ref{exact_sp3}}\label{sec:Proof_exactsp3}

We will consider the following cases separately; namely, $q < n_s$ and $
q \ge n_s$.

\subsubsection{$q < n_s$ Case}
We start by applying the identity (\ref{eq:DetProp}) to obtain the ergodic capacity, in the high SNR regime, as
follows
\begin{align}\label{eq:exactsp3_1} \left. {C\left( \rho
\right)} \right|_{\alpha ,\rho  \to \infty ,\alpha /\rho  = \beta }
= \frac{1}{2}\left[ {q\log _2 \rho  - q\log _2 \left( {\frac{\beta
}{{n_s n_r }}} \right) + E\left\{ {\log _2 \det \left( {{\bf{\bar
LH}}_1 {\bf{\tilde H}}_1^\dag  } \right)} \right\}} \right] \; .
\end{align}
The high SNR slope can be calculated as
\begin{align}
\begin{array}{*{20}c}
   {S_\infty   = \frac{q}{2}} & {\rm bit/s/Hz\left( {3dB} \right)}  \\
\end{array}.
\end{align}
Applying (\ref{eq:highsnroffset}), the high SNR power offset is given by
\begin{align}\label{eq:highsnroffset1}
\mathcal{L}_\infty   = \frac{q}{2}\log _2 \left( {\frac{\beta }{{n_s
n_r }}} \right) - \frac{1}{2}E\left\{ {\log _2 \det \left(
{{\bf{\bar L \tilde H}}_1 {\bf{\tilde H}}_1^\dag  } \right)}
\right\} \; .
\end{align}
Invoking \textit{Theorem \ref{lnexpdetNew}} and simplifying yields the high SNR power offset
for case $q < n_s$.

The proof of (\ref{eq:poweroffset_sp0}) follows along similar lines
to that used above, but in this case invoking \textit{Theorem
\ref{lnexpdetNew}} in place of \textit{Theorem \ref{exact_sp3}}.


\subsubsection{$q \ge n_s$ Case} In the high SNR regime, the ergodic capacity can be approximated as
\begin{align}\label{eq:exactsp3_2}
\left. {C\left( \rho  \right)} \right|_{\alpha ,\rho  \to \infty
,\alpha /\rho  = \beta }  = \frac{1}{2}\left[ {n_s \log _2 \left(
\rho  \right) - n_s \log _2 \left( {\frac{\beta }{{n_s n_r }}}
\right) + E\left\{ {\log _2 \det \left( {{\bf{\tilde H}}_1^\dag
{\bf{\bar L\tilde H}}_1 } \right)} \right\}} \right] \; .
\end{align}
In this case, the high SNR slope is
\begin{align}
   {S_\infty   = \frac{n_s}{2}}  \hspace*{0.5cm} {\rm bits/s/Hz\left( {3dB} \right)}
\end{align}
and the high SNR power offset can be obtained as
\begin{align}\label{eq:highsnroffset2}
\mathcal{L}_\infty   = \frac{n_s}{2}\log _2 \left( {\frac{\beta
}{{n_s n_r }}} \right) - \frac{1}{2}E\left\{ {\log _2 \det \left(
{{\bf{\tilde H}}_1 {\bf{\bar L}} {\bf{\tilde H}}_1^\dag  } \right)}
\right\} \; .
\end{align}
The result follows by applying \textit{Theorem \ref{lnexpdetNew}}.

\subsection{Proof of Corollary \ref{upperbound_sp4}}\label{sec:Proof_Upperboundsp4}
Substituting $n_r = 1$ into (\ref{eq:upperbound}) yields
\begin{align}
C_U^{n_r  = 1} (\rho)  = \frac{1}{2}\log _2 \left( {a^{ - n_d }
\left[ {U\left( {n_d ,n_d  + 1,\frac{{1 + \rho }}{\alpha }} \right)
+ \rho n_d U\left( {n_d  + 1,n_d  + 1,\frac{{1 + \rho }}{\alpha }}
\right)} \right]} \right) \; .
\end{align}
Using the following properties of the confluent
hypergeometric function of the second kind \cite{Gradshteyn00}:
\begin{align}
U\left( {a,a,z} \right) = e^z z^{1 - a} E_a \left( z \right)
\end{align}
and
\begin{align}
U\left( {a,a + 1,z} \right) = z^{ - a},
\end{align}
we get the final expression for $C_U^{n_r  = 1}(\rho)$ in
(\ref{eq:upboundsp4}). Note that $C_U^{n_r  = 1}(\rho) $ can be
lower and upper bounded as
\begin{align}
C_{U,1}^{n_r  = 1}(\rho)  < C_U^{n_r  = 1}(\rho)  \le C_{U,2}^{n_r =
1}(\rho) ,
\end{align}
with
\begin{align} \label{eq:U1}
C_{U,1}^{n_r  = 1}(\rho)  = \frac{1}{2}\log _2 \left( {1 + \rho n_d
\frac{1}{{\frac{{1 + \rho }}{\alpha } + n_d  + 1}}} \right)
\end{align}
and
\begin{align} \label{eq:U2}
C_{U,2}^{n_r  = 1}(\rho)  = \frac{1}{2}\log _2 \left( {1 + \rho n_d
\frac{1}{{\frac{{1 + \rho }}{\alpha } + n_d }}} \right),
\end{align}
where we have used the inequality \cite[Eq. 5.1.19]{Abramowitz74}.
Taking $n_d \to \infty$, we see that both (\ref{eq:U1}) and
(\ref{eq:U2}) converge to the same limit in (\ref{eq:Cusp4}). Taking
$\alpha \to \infty$ and ultilizing \cite[Eq. 5.1.23]{Abramowitz74}, we obtain (\ref{eq:Cusp41}).

\subsection{Proof of Corollary \ref{upperbound_sp1}}\label{sec:Proof_Upperboundsp1}

Note that when $ \rho  \to \infty $, then $ a \to 0 $.  Therefore,
we apply the following asymptotic first-order expansion for the
confluent hypergeometric function \cite{Abramowitz74}
\begin{align}
\begin{array}{*{20}c} \label{eq:HypGeoExpand}
   {U\left( {c,b,z} \right) = z^{ - c}  + o\left( 1 \right),} & {z \to \infty }
\end{array} \;
\end{align}
to yield the desired result.

\subsection{Proof of Theorem \ref{lowerbound}}\label{sec:Proof_Lowerbound}

We will use the lower bound derived in \cite[Theorem 1]{Oyman03} and
consider the following cases separately; namely, $q < n_s$ and $ q
\ge n_s$.

\subsubsection{$q < n_s$ Case}

Applying the (\ref{eq:DetProp}) and \cite[Theorem 1]{Oyman03} to (\ref{eq:ergcapa}), we lower bound the
ergodic capacity, conditioned on $\bf{L}$, as follows
\begin{align}\label{eq:lowerbound1}
C\left( \rho  \right) \ge q\log _2 \left( {1 + \frac{{\rho \alpha
}}{{n_s n_r }}\exp \left( {\frac{1}{q}E\left\{ {\ln \det \left(
{{\bf{L\tilde H}}_1 {\bf{\tilde H}}_1^\dag  } \right)} \right\}}
\right)} \right) \; .
\end{align}
Now, using \textit{Theorem \ref{lnexpdetNew}} yields the desired
result.

\subsubsection{$q \ge n_s$ Case} In this case, the lower bound can
be written as
\begin{align}\label{eq:lowerbound2}
C\left( \rho  \right) \ge n_s\log _2 \left( {1 + \frac{{\rho \alpha
}}{{n_s n_r }}\exp \left( {\frac{1}{n_s}E\left\{ {\ln \det \left(
{{\bf{\tilde H}}_1 \bf{L}{\bf{\tilde H}}_1^\dag  } \right)}
\right\}} \right)} \right) \; .
\end{align}
Again, we use \textit{Theorem \ref{lnexpdetNew}} to obtain the
desired result. 

\subsection{Proof of Corollary \ref{lowerbound_sp3}}\label{sec:Proof_Lowerboundsp3}

When $n_s \to \infty$, $\psi \left( {n_s  - q + k} \right)$ can be
approximated as \cite[Eq. 6.3.18]{Abramowitz74}
\begin{align}\label{eq:psiapproximate}
\left. {\psi \left( {n_s  - q + k} \right)} \right|_{n_s  \to \infty }  &\approx \ln \left( {n_s  - q + k} \right)\nonumber\\
&\approx \ln n_s \; \; .
\end{align}
Substituting (\ref{eq:psiapproximate}) into (\ref{eq:lowerbound})
yields the desired result.

\subsection{Proof of Corollary \ref{lowerbound_sp4}}\label{sec:Proof_Lowerboundsp4}
Taking $n_s \to \infty$ and using \cite[Eq. 6.3.18]{Abramowitz74}, we get (\ref{eq:lsp4ns}).

For the case $n_d \to \infty$, we first apply \cite[Eq. 5.1.19]{Abramowitz74} and \cite[Eq. 8.365.3]{Gradshteyn00} to obtain the following approximation
\begin{align}\label{eq:lbsp4app1}
\exp \left( {\frac{{1 + \rho }}{\alpha }} \right)\sum\limits_{l =
1}^{n_d - 1} {E_{l + 1} \left( {\frac{{1 + \rho }}{\alpha }} \right)} \approx \psi \left( {n_d  + \frac{{1 + \rho }}{\alpha }} \right) - \psi \left( {\frac{{1 + \rho }}{\alpha }} \right) \; .
\end{align}
Furthermore, substituting (\ref{eq:lbsp4app1}) into (\ref{eq:lowerboundsp4}) and using \cite[Eq. 8.365.5]{Gradshteyn00} and \cite[Eq. 6.3.18]{Abramowitz74} yields (\ref{eq:lsp4nd}).

Now consider the case $\alpha \to \infty$. Utilizing the recurrence relation for the exponential integral
\cite[Eq. 5.1.14]{Abramowitz74}, the summation in (\ref{eq:lowerboundsp4}) can be alternatively written as
\begin{align}\label{eq:recurrsum}
&\exp \left( {\frac{{1 + \rho }}{\alpha }} \right)\sum\limits_{l =
1}^{n_d - 1} {E_{l + 1} \left( {\frac{{1 + \rho }}{\alpha }} \right)}
\nonumber \\
& \hspace*{2cm} = \exp \left( {\frac{{1 + \rho }}{\alpha }} \right)E_1 \left( {\frac{{1 + \rho }}{\alpha }} \right) + \sum\limits_{l = 1}^{n_d - 1} {\frac{1}{l}\left[ {1 - \frac{{1 + \rho }}{\alpha }\exp \left( {\frac{{1 + \rho }}{\alpha }} \right)E_l \left( {\frac{{1 + \rho }}{\alpha }} \right)} \right]} \nonumber \\
& \hspace*{2cm} = \exp \left( {\frac{{1 + \rho }}{\alpha }}
\right)\left[ {E_1 \left( {\frac{{1 + \rho }}{\alpha }} \right) -
\sum\limits_{l = 1}^{n_d - 1} {\left( {\frac{{1 + \rho }}{{\alpha l}}}
\right)E_l \left( {\frac{{1 + \rho }}{\alpha }} \right)} } \right]+
\psi \left( {n_d } \right) + \gamma 
\end{align}
where $\gamma  = 0.577215 \ldots $ is the Euler's constant. Note
that, in deriving (\ref{eq:recurrsum}), we have applied the
definition of the digamma function \cite[Eq. 8.365.4]{Gradshteyn00}.
Using the series expansion given in \cite[Eq. 5.1.11]{Abramowitz74},
when $\alpha \to \infty$, we get
\begin{align}\label{eq:recurrsum3}
\left. {E_1 \left( {\frac{{1 + \rho }}{\alpha }} \right)}
\right|_{\alpha  \to \infty }  \to  - \gamma  - \ln \left( {\frac{{1 +
\rho }}{\alpha }} \right) \;
\end{align}
and therefore
\begin{align}\label{eq:recurrsum4}
\left. {\sum\limits_{l = 1}^{n_d - 1} {\left( {\frac{{1 + \rho
}}{{\alpha l}}} \right)E_l \left( {\frac{{1 + \rho }}{\alpha }}
\right)} } \right|_{\alpha  \to \infty }  \to 0 \; .
\end{align}
Applying (\ref{eq:recurrsum})--(\ref{eq:recurrsum4}) in
(\ref{eq:lowerboundsp4}) yields the desired result.

\subsection{Proof of Corollary \ref{lowerbound_sp1}}\label{sec:Proof_Lowerboundsp1}
Using the following approximation \cite{Abramowitz74}
\begin{align}
\begin{array}{*{20}c}
   {E_v \left( z \right) \approx \frac{1}{z}e^{ - z} \left( {1 + o\left( {\frac{1}{z}} \right)} \right)} & {\left| z \right| \to \infty }
   \\
\end{array} \; ,
\end{align}
$\varsigma_{m+n}(a)$
can be approximated as
\begin{align}
\left. \varsigma_{m+n}(a)
\right|_{\rho  \to \infty }
\approx \Gamma \left( {\tau  - 1} \right)\psi \left( {\tau  - 1}
\right) \; ,
\end{align}
which leads to the final result.




\end{document}